\documentclass[format=acmsmall, review=false]{acmart}
\usepackage{acm-ec-25}
\usepackage{booktabs} % For formal tables
\setcitestyle{acmnumeric}
\usepackage{caption}
\usepackage{subcaption}
\usepackage[utf8]{inputenc} 
\usepackage{booktabs}
\usepackage{enumitem}
\usepackage{xcolor}
\usepackage{bbm}
\usepackage{algorithm}
\usepackage{algpseudocodex}
\usepackage[hyperref]{}
\usepackage{makecell}
\usepackage{multirow}
\usepackage{amsthm}
\usepackage{paralist}
\usepackage{bm}
\usepackage{cleveref}
\usepackage{graphicx}
\usepackage{subcaption} %george added
\usepackage{url} %george added
\usepackage{tabularx}
\usepackage[symbol]{footmisc}

\newtheorem{theorem}{Theorem}
\theoremstyle{definition}
\newtheorem{example}{Example}[section]
\newtheorem{proposition}{Proposition}

\newcommand{\mBeta}{\texttt{mBeta}\xspace}
\newcommand{\Beta}{\texttt{Beta}}

\newcommand{\Dir}{\texttt{Dir}}
\newcommand{\cov}{\texttt{cov}}
\usepackage{xspace}

\newcommand{\USW}{\texttt{USW}\xspace}
\newcommand{\NSW}{\texttt{NSW}\xspace}
\newcommand{\ENVY}{\texttt{Envy}\xspace}
\newcommand{\MMS}{\texttt{MMS}\xspace}

\newcommand{\N}{\mathbb{N}}
\newcommand{\R}{\mathbb{R}}
\newcommand{\Ind}{\mathbbm{1}}
\newcommand{\E}{\mathbb{E}}

\newcommand{\UMass}{University of Massachusetts Amherst}
\title{Deploying Fair and Efficient Course Allocation Mechanisms}
\date{}
\author{George Bissias}
\affiliation{
  \institution{University of Massachusetts Amherst}
  \city{}
  \country{USA}
}
\author{Cyrus Cousins}
\affiliation{%
  \institution{Duke University}
  \country{USA}
}
\authornote{Work conducted while author was a postdoctoral research fellow at the University of Massachusetts, supported by a CDS Fellowship}
\author{Paula Navarrete D{\'i}az} 
\affiliation{%
  \institution{University of Massachusetts Amherst}
  \country{USA}
}
\authornote{Lead Author}
\author{Yair Zick}
\affiliation{%
  \institution{University of Massachusetts Amherst}
  \city{}
  \country{USA}
}
\begin{abstract}
    Universities regularly face the challenging task of assigning classes to thousands of students while considering their preferences, along with course schedules and capacities. 
    Ensuring the effectiveness and fairness of course allocation mechanisms is crucial to guaranteeing student satisfaction and optimizing resource utilization. 
    We approach this problem from an economic perspective, using formal justice criteria to evaluate different algorithmic frameworks. 
    To evaluate our frameworks, we conduct a large scale survey of university students at \UMass, collecting over 1000 student preferences. This is, to our knowledge, the largest publicly available dataset of student preferences.  
    We develop software for generating synthetic student preferences over courses, and implement four allocation algorithms: the serial dictatorship algorithm used by \UMass; Round Robin; an Integer Linear Program; and the Yankee Swap algorithm. 
    We propose improvements to the Yankee Swap framework to handle scenarios with item multiplicities. 
    Through experimentation with the Fall 2024 Computer Science course schedule at \UMass, we evaluate each algorithm's performance relative to standard justice criteria, providing insights into fair course allocation in large university settings.
\end{abstract}
\begin{document}
\begin{titlepage}

\maketitle

% Optionally include a table of contents
% \vspace{-0.37cm}
% \setcounter{tocdepth}{2} % adjust to 1 if desired
% \tableofcontents

\end{titlepage}
\section{Introduction}
%Always start with an example
Public universities regularly run large-scale matching markets: enrolling students to classes.
There are over 15 million undergraduate students in the United States alone \cite{nces2023ugradenrollment}, with some universities boasting cohorts numbering in the tens of thousands \cite[Table 312.10]{nces2023ugradenrollment}.   

The sheer scale of the course allocation problem requires the use of automated assignment mechanisms, which typically 
\begin{inparaenum}[(a)]
    \item collect student preferences and 
    \item assign students to classes based on their eligibility/preferences/priority. 
\end{inparaenum}
%This is the problem
Ideally, course allocation mechanisms should be fast, effective and satisfy certain objectives. 
One reasonable objective is \emph{efficiency}: ensuring that classes are assigned to students who actually want and are able to take them; another is \emph{fairness}: ensuring that class seats are fairly distributed among students. 
Other constraints must be accounted for as well: student schedules should be \emph{feasible}, i.e., the mechanism does not cause a scheduling conflict; another is \emph{student priority}: course allocation mechanisms should prioritize students with greater needs. 
%Why is this an important problem?
%complete here
Properly designed, course allocation mechanisms guarantee student satisfaction, facilitate timely graduation, and optimize the utilization of public resources. 
%What is your proposed solution? 

The course allocation problem can be naturally modeled as a problem of \emph{fair allocation of indivisible goods} \cite{brandt2016handbookfairdiv}: students are \emph{agents} who exhibit preferences over \emph{items} --- seats in classes. 
Our goal is to design a \emph{mechanism} that assigns items to agents while satisfying certain design criteria. 
\emph{Prima facie}, the course allocation problem is well-positioned to be central in the fair allocation community: 
it is a regularly occurring market involving thousands of agents with preferences over thousands of items (unlike smaller domains, e.g., those offered by the Spliddit platform \cite{goldman2015spliddit}); moreover, it offers a rich landscape of constraints and modeling challenges: modeling student preferences, accounting for scheduling clashes, and determining acceptable justice criteria. 
\emph{The lack of public datasets} is one major obstacle to the proper scientific analysis of course allocation mechanisms. 
Recent works, e.g., the Course Match mechanism,  \cite{budish2011ef1,budish2017coursematch,budish2023courseallocation,soumalias2024courseallocation} offer empirical analysis of course allocation instances; however, these works do not publicly release their datasets, nor are their proposed allocation mechanisms publicly available. 
This makes it difficult to independently verify the efficacy of their proposed approach.

In this work, we take a first step towards creating a \emph{publicly available}, \emph{large-scale} ecosystem for the analysis of the course allocation problem.

\subsection{Our Contribution}
We implement several algorithmic frameworks for large-scale fair allocation, and test them on preference data collected from students at \UMass. The data and algorithmic frameworks are all in public repositories (data: \url{https://github.com/Fair-and-Explainable-Decision-Making/course-allocation-data}; frameworks: \url{https://github.com/Fair-and-Explainable-Decision-Making/yankee-swap-allocation-framework}.
In more detail, our contributions include:  
\begin{description}[leftmargin= 0cm]
    \item[Data collection:] we collect preference data from 1061 students in the computer science department of \UMass (see \Cref{sec:data} for details). In the fair allocation domain, this dataset represents one of the largest sources of publicly available preference data.
    % This dataset is one of the largest sources of publicly available data on agent preferences in the fair allocation of indivisible goods domain. 
    \item[Software for Large-Scale Fair Allocation:]  
    we implement four allocation mechanisms (\Cref{sec:allocation_algs}): Round Robin (a classical scheduling algorithm, also known as the draft mechanism \cite{budish2012roundrobincourseallocation}), a version of Serial Dictatorship (the algorithm used at \UMass), an Integer Linear Program that finds the solution maximizing utilitarian social welfare, and Yankee Swap \cite{viswanathan2023general}, a recently introduced fair allocation mechanism. 
    We discuss practical considerations in the development of the mechanisms, as well as methods to utilize structural properties of the course allocation problem which significantly improve the runtime of the Yankee Swap framework (\Cref{sec:yswithmult}). 
    % Our software is open-source, allowing the research community to experiment with implementations of large-scale fair allocation mechanisms.
We test each algorithm against the Fall 2024 \UMass Computer Science course schedule, and measure its performance relative to a number of standard \emph{justice criteria} (\Cref{sec:experiments}). 
    \item[Synthetic Student Preference Generation:] to test our algorithmic frameworks under different regimes of course supply and demand, we introduce a synthetic student generator which creates new students based on the collected preference data. The student preference generator creates student profiles that, while distinct from the original data, follow a similar preference distribution (\Cref{sec:synth_students}). 
\end{description}

\subsection{Related work}
% \paragraph{Deploying fair allocation mechanisms}
% \paragraph{Open Source Datasets}
% \paragraph{Fair Course Allocation Mechanisms}
The course allocation problem is well-studied in combinatorial optimization, with various mechanisms proposed to balance efficiency, fairness, and strategic considerations. 
\citet{hatfield2009serialdictatorship} shows that serial dictatorships are the only Pareto efficient, strategyproof and non-bossy mechanisms; \citeauthor{hatfield2009serialdictatorship} focuses on additive preferences, whereas we study more general submodular utilities. 
\citet{budish2012roundrobincourseallocation} examine the Draft (Round-Robin) mechanism for course allocation. 
\citeauthor{budish2012roundrobincourseallocation} show that the Draft mechanism retains high efficiency guarantees in practice, despite students misreporting their preferences. \citeauthor{budish2012roundrobincourseallocation} also collect data on a similar scale to ours (900 students at Harvard Business School), with students ranking classes on a Likert scale. 
The bidding point mechanism is widely used for course allocation \cite{sonmez2010course}: students distribute a fixed number of points among their preferred courses. 
Bids are then processed in decreasing order and honored if the course and the student's schedule are not at capacity. This approach focuses solely on efficiency, neglecting equity, and is highly susceptible to strategic manipulation, as students may misreport their preferences to get better outcomes \cite{atef2020optimization}. 
Furthermore, since bids serve as proxies for preferences, the inferred preferences may not accurately reflect true student demand, leading to market distortions \cite{sonmez2010course, krishna2008research}. 

\citet{sonmez2010course} propose a Gale-Shapley Pareto-dominant market mechanism that separates the dual role of bids by assigning courses through a matching mechanism considering students preferences, while using bids solely to break ties. In any case, inefficiencies arise from students overbidding on courses they could have gotten with fewer points. 
\citet{atef2020optimization} introduce multi-round algorithms, based on matching and second-price concepts, which addresses these inefficiencies and show promising empirical results. 

\citet{diebold2014course} study course allocation as a two sided matching problem, where courses have preferences over students, induced by student priority. 
They analyze the Gale-Shapley Student-Optimal Stable mechanism --- a modified version of the Gale-Shapley deferred acceptance algorithm \cite{galeshapley} --- and the Efficiency Adjusted Deferred Acceptance mechanism \cite{kesten2010schoolchoice}, which reduces the welfare losses produced by the deferred acceptance algorithm, but is not strategyproof. 
Similar to other approaches \cite{nogareda2017optimizing, romero2024strategic}, these fail to incorporate conflicts between classes. %Handling scheduling conflicts or course sections is a crucial aspect of course allocation. Some works address this issue by modeling the problem using interval graphs \cite{biswas2023algorithmic} and conflict graphs \cite{biswas2024fair}.

\citet{budish2011ef1} introduces the Approximate Competitive Equilibrium from Equal Incomes (A-CEEI) mechanism, which is strategy-proof at large, and approximates envy-freeness and maximin share fairness, both in theory and empirically (see also \cite{othman2010courseallocation}). 
In the context of course allocation, this mechanism faces some challenges. 
A-CEEI assumes that all students have approximately equal budgets and, therefore, equal priority. Due to academic performance, seniority, or graduation requirements, some students may require higher priority when enrolling in courses. \citet{kornbluth2021undergraduate} propose the Pseudo-Market with Priorities mechanism, which incorporates competitive equilibrium with fake money while accounting for student priorities. 
More importantly, A-CEEI approximation errors can result in capacity constraint violations. In addition, early implementations faced scalability issues \cite{budish2017coursematch, atef2020optimization}.

Course Match \cite{budish2017coursematch}, an A-CEEI-based approach, addresses these issues by introducing two additional stages to ensure feasibility, together with an interface for students to report utilities, and a software system capable of handling reasonably sized instances. 
Course Match relies on students reporting their complete ordinal preferences over all possible course bundles --- an impractical requirement. Instead, students only rate individual courses and course pairs. 
At the same time, Course Match does not guarantee optimal outcomes if students misreport or fail to articulate their true preferences. To mitigate this issue, \citet{soumalias2024courseallocation} propose a machine-learning-powered version of Course Match to reduce reporting errors by iteratively asking student to rate personalized pairwise comparison queries. 

Despite these efforts, Course Match is not guaranteed to find a price vector and corresponding allocation that respects course capacities, and may not run in reasonable time for large instances \cite{budish2017coursematch}. 
Furthermore, Course Match is a commercial product. 
To our knowledge, there is currently no publicly available implementation of Course Match, nor is there large-scale empirical data to evaluate its effectiveness in real-world settings. 
Lastly, Course Match places a significant elicitation burden on students. 
Students are encouraged to study a 12-page manual in order to understand the mechanism and enroll in courses \cite{coursematchmanual2020}. 
While this may be reasonable for MBA students or those in advanced economics-related programs, it may not be suitable for undergraduates from diverse majors and academic backgrounds at a large public university. 

\citet{biswas2023algorithmic} analyze the complexity of finding fair and efficient course allocations under course conflicts (see also \cite{biswas2024fair}). They establish that the problem is computationally intractable, unless the number of agents is small. At \UMass, course schedules are very standardized, which makes course schedule conflicts well-structured and tractable. 

Fair allocation algorithms and solution concepts have been extensively studied in recent years, with a particular focus on leveraging the structure of agent preferences (see survey by \citet{aziz2022fairallocationsurvey}). 
Recent works focus on settings where agents have \emph{binary submodular} utilities, i.e., items have a marginal utility of either $0$ or $1$, and agents exhibit diminishing marginal returns on items as their bundles grow \cite{babaioff2021EF,barman2021matroid,garg2020mms,benabbou2021MRF,viswanathan2023yankeeswap,viswanathan2023general}.
Following these works, we implement the Yankee Swap algorithm \cite{viswanathan2023yankeeswap} as one of our benchmarks. 

The Yankee Swap algorithm outputs a \emph{Lorenz dominating} allocation, which is guaranteed to satisfy several theoretical guarantees \cite{babaioff2021EF,benabbou2021MRF}. In addition to maximizing the minimal utility of any agent (which is implied by the leximin property), it is approximately envy-free, maximizes utilitarian welfare, and offers each agent at least half of their maximin share guarantee. Furthermore, Yankee Swap is a strategyproof mechanism, and can encode student priorities \cite{viswanathan2023general}. Our empirical evaluation confirms the effectiveness of the Yankee Swap framework, providing evidence for its potential efficacy as a practical course allocation mechanism.
% Since understanding the algorithm is essential for students to accurately report their true preferences, this complexity could significantly undermine the guarantees that Course Match offers.

\section{Course Preference Data}
\label{sec:data}
We conducted a survey to gather student preferences for courses offered by the Computer Science department at \UMass during the Fall 2024 semester. We reached out to students in May 2024, shortly after the final examination period but before grades were released. 
This timing was selected for two reasons: first, while classes were over, most students were still on campus. Secondly, the survey was released shortly after students indicated their course preferences for the Fall 2024 on the \UMass portal. 
Thus, the survey timing ensured that most students were available --- they were not overly busy with classwork or studying for exams, nor were they off-campus and enjoying their summer break --- and had time to consider their preferences before answering our survey --- they had recently finalized their Fall 2024 schedule.  
The survey achieved a significant effective response rate of 30.33\%, with a relatively balanced response rate across students of different academic status (see \Cref{table:student_quantities} and discussion). 
% We also obtained data on the Fall 2024 computer science course capacities and schedules at \UMass.

We first describe the survey and data collection process, followed by our methodology for generating synthetic students.
\subsection{Student Preference Survey}
\label{sec:survey}
The Fall 2024 course schedule for the Computer Science department includes 96 distinct course sections. 
Each section is identified by a combination of attributes such as catalog number, section number, course name, instructor, capacity, and schedule. 
The catalog number indicates the course level: numbers ranging from 100 to 400 are undergraduate courses aimed at students at different stages of their undergraduate studies; numbers in the 500 to 600 range are graduate-level courses, primarily aimed at MS and PhD students, although they are sometimes open to senior undergraduate students. Multiple sections may offer the same catalog number. However, the combination of catalog number and section number is unique, and is referred to as a \emph{course} from now on. 

We invited all undergraduate and graduate students in the Computer Science department to participate in the survey. Participation was voluntary, and participants received a \$10 Amazon gift card. 
The survey was hosted on Qualtrics, an online survey platform, and consisted of the following questions:
\begin{description}[leftmargin=0cm]
\item[Current Status:] current academic status refers to the student's current year in their program, whether they are an undergraduate (freshman, sophomore, junior, or senior) or a graduate student (enrolled in a master's or Ph.D. program). We used this information to tailor the courses presented in surveys to reduce cognitive load. 
For example, graduate students were not asked to rank undergraduate-level courses, and first-year students were not asked to rank graduate-level courses. 
More importantly, student status determines students' priority when enrolling in classes. 
PhD students are given priority over MS students, who are given priority over seniors, and so on. 
\item[Desired course load:] we asked students how many courses they wish to enroll in before and after the add/drop period. Since students are often uncertain about their final preferences, many initially over-enroll to explore their options before finalizing their schedules. Therefore, these quantities might not match. 
\item[Preferred course categories:] students indicated areas of interest: Theory, Systems, AI/ML, and Other. This information was collected to filter out courses that fall outside a student’s stated interests, reducing their cognitive load when ranking classes. 
Additionally, understanding students' preferred categories offers valuable insights into their academic focus, which can further inform course allocation decisions and guide future curriculum planning.
\item[Undesired time slots:] students indicated if there were any particular time slots they were not able to take classes in. 
This information can be used to ensure that students are not enrolled in courses scheduled at times they are unable to attend. When testing our methods, we do not utilize this information in order to obtain a richer (and more competitive) preference landscape, but one can simply encode every class that conflicts with the individual student schedule as an undesirable class.
\item[Preference Over Courses:] after the initial screening questions, students were shown a tailored list of courses based on their academic status and preferred course categories. 
We presented students with the following key details --- course name, catalog number, section, instructor, day and time --- in order to help students make informed decisions about their preferences. 
For each course listed, students rated their interest on a scale from 1 (not interested) to 7 (very interested), with an additional option to select 8 (required) if the course was a requirement for their program. All courses were rated 1 by default. 
\end{description}
We used the initial questions to reduce the cognitive load on students: we did not collect student preferences on any courses that were inappropriate for their academic status or did not match their area of interest. 
As a result, the number of courses students were requested to rank was significantly reduced.
\begin{table}
    \centering
    \small
    \begin{tabularx}{\textwidth}{cXXXXXX}
    \toprule
         Academic Status& Freshmen &Sophomore & Junior & Senior & MS & PhD\\
         \toprule
         Avg. \# of presented classes& 29 & 67 & 71 & 82 & 26 & 27\\
         \bottomrule
    \end{tabularx}
    \caption{Average number of classes students were asked to rank after preliminary filtering, by academic status.}
    \label{tab:avg-num-courses-ranked}
\end{table}
As can be seen in \Cref{tab:avg-num-courses-ranked}, Freshmen and MS/PhD students had the greatest reduction in the number of classes they observed, while juniors and seniors were still presented with a large number of classes. 
This is mostly due to the fact that juniors and seniors can take classes at all levels, whereas graduate students and freshmen are offered a more limited selection. 

We did not allow students to express preferences over bundles of courses. 
This is a departure from other preference elicitation models in the course allocation domain, e.g., \cite{budish2017coursematch,soumalias2024courseallocation}. 
We elicited information on individual courses since it allows students to express preferences over a large number of courses while avoiding significant cognitive overload. 
In addition, our elicitation protocol maps well to constraint based utility models: we can naturally map the results of our survey to \emph{submodular} (or nearly submodular) student preference functions. 
Having a complete picture of student preferences over individual courses also allows us to naturally simulate additional student profiles (as described in \Cref{sec:synth_students}). 
% Finally, we offered students the opportunity to comment on the survey
%cognitive burden
%submodularity
%comments are allowed
%will explore methods in future work

The total number of survey responses, the actual number of students (provided by the department), and the corresponding response rates --- both overall and broken down by academic status --- are summarized in \Cref{table:student_quantities}. 
256 respondents did not specify their academic status, and as a result, we did not use their recorded rankings in our analysis. 
109 respondents indicated their status but submitted empty preferences (those students were not paid). 
We utilized 700 effective responses with both academic status and non-empty preferences, achieving an effective response rate of 30.33\%. 
\Cref{table:student_quantities} provides a detailed breakdown of these values by academic status.

\begin{table}
\centering
\begin{tabularx}{0.9\textwidth}{cccccc} 
% \hline
\toprule
 Academic  & Number of  & Number of & Response & Effective  & Effective Response \\
  Status &  Students &  Responses & Rate (\%) & responses &   Rate (\%) \\
\toprule
Freshmen & 239 & 156 & 65.27 &125 & 53.30 \\
\midrule
Sophomore & 327 & 134 & 40.98 &113 & 34.56 \\
\midrule
Junior & 408 & 143 & 35.05  &126 & 30.88\\
\midrule
Senior & 573 & 135 & 23.56 &117 & 20.42\\
\midrule
MS & 613 & 190 & 31.00 &172 & 28.06 \\
\midrule
PhD  & 148 & 51 & 34.46 & 47 & 31.76\\
\midrule
Unspecified & N/A & 256 & N/A & 0 & N/A\\
\bottomrule
Total & 2308 & 1065 & 46.14 & 700 & 30.33\\
\bottomrule
\end{tabularx}
\caption{Number of students, number of responses, and response rates, both overall and broken down by academic status. Effective response rates reflect responses with non-empty preferences and specified academic status.}
\label{table:student_quantities}
\end{table}
\subsection{Synthetic Students}
\label{sec:synth_students}
Let us next describe our methodology for generating synthetic student preference profiles.
We let $S = \{1,\dots,s\}$ be the set of respondents and $G = \{g_1,\dots,g_m\}$ be the set of courses. 
For each course $g \in G$, respondent $i\in S$ provides a response in the range $\{1, \ldots, 8\}$ (where $8$ corresponds to a required class). 
We normalize these responses to a preference vector for respondent $i$, $\vec \theta^i \in [0,1]^m$. 
Fundamentally, we assume that each respondent will randomly formulate binary preferences over the set of classes such that the normalized value $\theta^i_g \in [0,1]$ represents the probability that respondent $i$ wants to take the class $g$.

%We accomplish this end by means of a Bayesian inference procedure, which is the multivariate analog to the following scalar process. Imagine we wish to randomly generate a family of Bernoulli distributions similar, but not identical, to a single Bernoulli distribution having mean $p$. In this problem, we might begin by generating a sequence of binary observations from $\texttt{Bern}(p)$. The conjugate prior for the Bernoulli distribution is the beta distribution, which means that if we choose a beta as our prior, then the posterior is also beta. Thus, we may conduct Bayesian inference conditioned on the observations to derive posterior distribution $\texttt{Beta}(\alpha, \beta)$, where $\alpha$ and $\beta$ are the associated hyperparameters. And we can sample a new, similar Bernoulli distribution by drawing $p' \sim \texttt{Beta}(\alpha, \beta)$. Naturally, the more observations that are drawn for the inference process, the closer each $p'$ drawn will be to $p$. In this way it is possible to control the similarity of distributions in the family.

Using the preference vector $\vec \theta^i$, we define an inference procedure that produces a distribution over preference vectors that are likely similar, but not equal, to $\vec \theta^i$. 
The process begins by sampling binary vectors $\ell$ times according to probabilities $\vec \theta^i$, which form a data matrix $\mathcal{D}^i \in \{0,1\}^{\ell \times m}$; if we take the limit $\ell \rightarrow \infty$, $\mathcal{D}^i$ uniquely describes $\vec \theta^i$.
Using the matrix $\mathcal{D}^i$, we infer a \emph{multivariate beta} posterior distribution, $\mBeta(\vec\gamma^i)$, as defined in \citet{westphal2019simultaneous}. 
The vector $\vec \gamma^i \in [0,1]^{2^m}$ is an implicit parameter that captures dependencies among course preferences. 
In order to avoid combinatorial blowup, the inference process is adjusted so that these dependencies are limited to covariances.  
We refer to the random vector $\vec \vartheta^i \sim \mBeta(\vec\gamma^i)$ as the $i$-th \emph{random student} and the $j$-th realization $\vec \sigma^{ij} 
\in [0,1]^m$, drawn from $\mBeta(\vec \gamma^i)$, as a \emph{synthetic student}. As $\ell$ increases, each synthetic student $\vec \sigma^{ij}$ will differ less and less from $\vec \theta^i$, the original preference vector for the respondent.
We defer further details of the inference process to Appendix~\ref{sec:synth_details}.

% \paragraph{Limitations.} There are several limitations to our model. 
% \begin{inparaenum}[(i)]
% \item In order to avoid a combinatorial blowup in the parameter space, it is necessary to restrict dependencies to covariance.
% \end{inparaenum}

%\paragraph{Kernel density estimation with \mBeta.}

%Consider a sequence of $s$ students: $\bm \theta^1, \ldots, \bm \theta^s$. For each student $\bm \theta^k$, construct a sequence of $t$ random data matrices $\bm Y^{k1}, \ldots, \bm Y^{kt}$ where for matrix $l$ we have $Y^{kl}_{ij} = \texttt{Bern}(\theta^k_j)$. A sample from $\bm Y^{kl}$ gives the data matrix $\bm{\mathcal{D}}^{kl}$. If we also fix course $j$, then we have a sequence of $n$ coin flips $\mathcal{D}^{kl}_{i1}, \ldots, \mathcal{D}^{kl}_{in}$ whose mean converges to $\theta^k_j$. %For student $k$ we denote the entire sample by $\mathcal{D}^k = \{ \bm{\mathcal{X}}^{k1}, \ldots, \bm{\mathcal{X}}^{kt} \}$. 

We model the population of students for each academic status (e.g. freshmen) using \emph{kernel density estimation} (KDE). %The \texttt{mBetaKDE} 
This distribution is constructed as a uniform mixture of the distributions $\mBeta(\vec \gamma^1), \ldots, \mBeta(\vec \gamma^s)$, and may be efficiently sampled  by drawing one of the kernels $\mBeta(\vec \gamma^i)$ uniformly at random, and then sampling from it. 
Because the number of rows $\ell$ in $\mathcal{D}^i$ determines the number of samples drawn from the marginal Bernoulli distribution $\texttt{Bern}(\theta^i_g)$ for each class $g$, $\ell$ can be seen as a smoothing parameter in much the same way that \emph{bandwidth} controls kernel smoothing in conventional KDE.  

Survey respondents also specify a \emph{course max}, or the maximum number of courses they are interested in taking. For each academic status, we model the course max preferences independently from course preferences. We first fit a multinomial distribution to the set of course max preferences among respondents for that status. Then, for each synthetic student $\vec \sigma^{ij}$ drawn according to $\vec \vartheta^i$, we independently draw a course max preference from the multinomial distribution.

\section{Course Assignment as a Fair Allocation Problem}
\label{sec:courseallocation}
%We implement a general suite of fair allocation mechanisms. (full access to the code base will be available upon publication).
% \footnote{\url{https://github.com/gbiss/2023-project-fair}}
We start by presenting the fair allocation of indivisible goods framework. 
Our model offers a slight departure from the standard problem formulation \cite{brandt2016handbookfairdiv}, since we allow item \emph{types}. 
Let $\N_0$ be the set of non-negative integers.
Let $N=\{1,2,...,n\}$ be a set of \emph{agents} and $G=\{g_1,g_2,...,g_m\}$ be a set of \emph{item types}. 
Each item type $g\in G$ has a limited number of identical copies $q(g)$; 
$\vec q=(q(g_1),...,q(g_m))$ is the vector describing the quantities of each item. 
In the context of course allocation, agents are students, item types are courses, and $q(g)$ is the number of available seats in course $g$. 

The classic fair allocation literature assigns agents \emph{bundles} of items, i.e., subsets of the set $G$. 
Since our framework allows items to have multiple copies, a bundle $S$ is a vector in $\N_0^m$, where $S_g$ is the number of copies of item $g$ that is in the bundle $S$. 
One can simply treat multiple copies of an item as distinct items and do away with the duplicate items model. However, the notion of item copies is useful for proving better upper bounds on the runtime of our proposed allocation mechanisms.
From a modeling perspective, treating individual course seats as distinct items makes running even somewhat simple allocation mechanisms like Round Robin impractical: instead of having $\sim 100$ item types, we need to store information about $\sim 10k$ individual items. 
An \emph{allocation} $X=(X_0, X_1, \dots, X_n)$ is a partition of item copies; each agent $i\in N$ is assigned a bundle $X_i$, and $X_0$ denotes the set of unassigned copies.  
The allocation $X$ is represented as a matrix, where $X_i$ is a vector in $\N_0^m$ that represents agent $i$'s bundle under allocation $X$. Since agents may own multiple copies of an item, $X_{i,g}$ is the number of copies of item $g$ in agent $i$'s bundle; 
however, in the course allocation domain students do not receive more than one seat in any class. 
An allocation is \emph{valid} if for any item $g\in G$, $X_{0,g}+ \sum_{i=1}^n X_{i,g} = q(g)$: the number of unassigned copies of $g$, $X_{0,g}$, plus the number of copies assigned to all agents equals exactly the number of copies $q(g)$. 
For ease of readability, for an allocation $X$ and an item type $g$, we say that $g\in X_i$ if $X_{i,g}>0$. 
Let $\Ind_g \in \{0,1\}^m$ be the indicator vector of item $g$, i.e., it is equal to $1$ in the $g$-th coordinate and is $0$ elsewhere. Given a bundle of items $S$, we write $S+g$ and $S-g$ to refer to $S+\Ind_g$ and $S-\Ind_g$, i.e., $S$ with an additional copy of the item $g$, or with one copy of $g$ removed.

Each agent $i \in N$ has a \emph{valuation function} $v_i:\N_0^m\to \N_0$ which depends only on the bundle $X_i$ allocated to $i$. 
We define the \emph{marginal utility} of agent $i$ from receiving an additional copy of the item type $g$ as 
$$\Delta_i(X_i, g)\triangleq v_i(X_i +g)-v_i(X_i ).$$

We say that $v_i$ is a \emph{binary} valuation if for any bundle $S\in \N_0^m$ and any item type $g$, $\Delta_i(S, g)\in \{0,1\}$.

Given two bundles $S,T \in \N_0^m$, we write $S\preceq T$ if for all $g\in G$, $S_g\leq T_g$; this is equivalent to stating that the bundle $S$ is a subset of $T$ in the standard fair allocation model; we write $S \prec T$ if any of these inequalities is strict. 
We say that $v_i$ is a \emph{submodular} function if for any two bundles $S,T \in \N_0^m$ such that $S\preceq T$, and any item type $g$, $\Delta_i(S, g)\geq \Delta_i(T, g)$. 
Intuitively, the more items an agent $i$ owns, the less marginal benefit they receive from additional items. 

Finally, a bundle $S\in \N_0^m$ is \emph{clean} \citep{benabbou2021MRF} with respect to agent $i \in N$ if for any bundle $T\prec S$ we have $v_i(S)> v_i(T)$, i.e., removing any item from $i$'s bundle strictly reduces their utility. 

\subsection{Encoding Student Valuations}
\label{subsec:valuation_functions}
% Thus far we have remained agnostic to the valuation function $v_i$ that defines the preferences of agent $i \in N$. From now on, we will consider students to have \emph{binary submodular} valuation functions over the bundles of courses. 
While our framework applies to general valuation settings, in our empirical evaluations we model students as having binary preferences. 
Binary preferences are extensively studied in the fair allocation literature, see e.g. \cite{babaioff2021EF,barman2018pathtransfers,barman2021matroid,benabbou2021MRF,benabbou2019group,halpern2020fairness,viswanathan2023yankeeswap,viswanathan2023general}. 
Thus, we have a good understanding of their properties and the types of guarantees we can offer; furthermore, there exist algorithmic frameworks that compute fair and efficient allocations for agents with binary submodular valuations. 
We chose to elicit numerical preferences from our survey participants in order to offer a more refined view of student preferences, and to support future empirical analysis of algorithmic frameworks beyond those that handle binary preferences. 
In other words, we see the \emph{data collection} and the \emph{algorithmic frameworks/preference encoding} portions of our work as complementary yet distinct contributions to the empirical evaluation of fair allocation mechanisms. 

Allocation mechanisms in the literature typically rely on oracle access to agent valuations. However, pre-computing and storing agent valuations in a fixed data table is intractable when dealing with thousands of students and hundreds of classes; 
what's more, the allocation mechanisms we utilize do not need to know agents' valuations for all possible bundles. 
To address this challenge, we encode student preferences via \emph{linear inequality constraints}. 
We can encode several types of natural course constraints via linear inequalities. For example, linear constraints can encode course scheduling conflicts, maximum enrollment caps for a student, substitutions (wanting to take course A or course B but not both), and taking only classes which the students approve.
% These constraints are illustrated in example \ref{exmp:linear_constraints}.

% \begin{exmp}
% \label{exmp:linear_constraints}
% Consider an example with five courses $\{g_1,\dots,g_5\}$, and a student $i$. 
% We assume that 
% \begin{inparaenum}[(a)]
% \item $g_1$ and $g_2$ have a scheduling conflict; 
% \item the student likes all courses but $g_3$; 
% \item they want to take $g_4$ or $g_5$, but not both and 
% \item they cannot be enrolled in more than three courses.
% \end{inparaenum} 
% Let $x_g^i$ be a binary variable which equals $1$ if and only if course $g$ is assigned to student $i$. Student $i$'s utility constraints are thus
% \begin{subequations}
% \begin{align}
% x_{1}+ x_{2} &\leq 1 \label{eq:lin_ineq_cons_1}\\
% x_{3} &\leq 0 \label{eq:lin_ineq_cons_2}\\
% x_{4}+ x_{5} &\leq 1 \label{eq:lin_ineq_cons_3}\\
% x_{1}+ x_{2}+x_{3}+ x_{4} +x_5 & \leq 3 \label{eq:lin_ineq_cons_4}
% \end{align}
% \end{subequations}
% where (\ref{eq:lin_ineq_cons_1}) ensures that she is not enrolled in two courses with conflicting schedules, (\ref{eq:lin_ineq_cons_2}) ensures that she is assigned only classes that she approves, (\ref{eq:lin_ineq_cons_3}) ensures she is not enrolled in two courses she does not wish to take simultaneously, and (\ref{eq:lin_ineq_cons_4}) ensures she complies with credit limits.
% \end{exmp}
We represent the preferences of a student $i$ by $r_i$ linear constraints through a $r_i \times m$ matrix $Z^i$ and a limit vector $\vec {b}^i$. A bundle $X_i$ is \emph{feasible}, i.e., it is clean with respect to agent $i$, if and only if
\begin{equation}
\label{eq:lin_ineq_cons}
    Z^i X_i \leq \vec{b}^i.
\end{equation}
These constraints can be computed as needed at a relatively low cost, offering a more manageable approach to modeling student preferences. 

Finally, the utility that agent $i$ obtains from a bundle $X_i$ is given by
\begin{equation}
\label{eq:lin_ineq_val}
    v_i(X_i) = \max \{ | S|: Z^i  S \leq \vec b^i, S \preceq X_i \},
\end{equation}
where  $|S|$ denotes the 1-norm of $S \in \N_0^m$. In other words, the utility an agent obtains from a bundle $X_i$ is the size of the largest feasible sub-bundle of $X_i$.

It is relatively straightforward to show that student preferences encoded by \Cref{eq:lin_ineq_val} are binary: adding a single item copy can increase the cardinality of the feasible set by at most one. 
% \begin{theorem}
% The valuation function $v_i: \N_0^m \to \mathbbm{R}$ is binary
% \end{theorem}
% \begin{proof}
% Let's consider an arbitrary bundle $X_i$. For any $g \in G$, since all sub-bundles of $X_i$ are also sub-bundles of $X_i+\mathbbm{1}_g$, then $v(X_i+\mathbbm{1}_g)\geq v_i(X_i),$
% and, therefore, $\Delta_i(X_i,g)\geq 0$.

% Let $g$ be an arbitrary item. Let $\vec x$ be the largest sub-bundle of $X_i$ such that $Z^i\vec x\leq \vec b^i$, and $\vec y$ the largest sub-bundle of $X_i+\mathbbm{1}_g$ such that $Z^i\vec y\leq \vec b^i$, then $v_i(X_i)=|\vec x|$ and $v_i(X_i+\mathbbm{1}_g)=|\vec y|$. If $y_g=0$, then $\vec y \preceq X_i$ and $|\vec y|\leq |\vec x|$. If $y_g=1 $, then $|\vec y|=|\vec y -\mathbbm{1}_g|+1$, where $y -\mathbbm{1}_g\preceq X_i$. Since $Z^i\vec y\leq \vec b^i$, then $Z^i(\vec y - \mathbbm{1}_g )\leq \vec b^i$ and $\vec y -\mathbbm{1}_g$ is a feasible sub-bundle of $X_i$, thus $|\vec y -\mathbbm{1}_g|\leq |\vec x|$. In any case, we have that $$\Delta_i(X_i,g) = v_i(X_i+\mathbbm{1}_g)-v(X_i)=|\vec y|-|\vec x|\leq 1$$
% Finally, since $v_i(X_i)$ is always an integer, we conclude that $\Delta_i(X_i,g)\in \{0,1\}$ for any $g\in G$. 
% \end{proof}
The Yankee Swap framework \cite{viswanathan2023general}, which we utilize in our empirical evaluation, offers several strong guarantees when agents have binary submodular valuations. While our encoding ensures that student valuations are binary, they are not necessarily submodular. For example, assume that a student $i$ likes classes $A$, $B$ and $C$. 
If class $A$ conflicts with $B$, $B$ conflicts with $C$, and $A$ does not conflict with $C$, then the marginal gain of giving $C$ to $i$ when they have $B$ is $0$; however, if they have both $A$ and $B$, then the marginal gain of $C$ is 1. 
In this scenario, items do not exhibit \emph{substitutability} \cite{paesleme2017gross}, a key property of submodular preferences. 
% In this scenario, classes $A$ and $C$ are \emph{synergistic}. Item synergies are antithetical to submodularity, where items can only exhibit \emph{substitutability} \cite{paesleme2017gross}.

% \begin{exmp}
% \label{exmp:not_submodular}
% Consider the same setting as in example \ref{exmp:linear_constraints}, but with the following additional constraint
% \begin{subequations}
% \begin{align}
% x_{1}+ x_{4}& \leq 1\label{eq:lin_ineq_cons_5} \tag{3e}
% \end{align}
% \end{subequations}
% This means that now $g_1$ conflicts with both $g_2$ (\ref{eq:lin_ineq_cons_1}) and $g_4$ (\ref{eq:lin_ineq_cons_5}), but $g_2$ does not conflict with $g_4$. Let's consider two bundles $\vec x$ and $\vec y$ such that: $x_1=1$ and 0 in every other component, and $\vec y=\vec x+\mathbbm{1}_2$. Then $\vec x\preceq \vec y$, and $\Delta(\vec x, g_4)<\Delta(\vec y, g_4)$, thus breaking submodularity.
% \end{exmp}However, every binary submodular value function can be represented by a set of linear inequality constraints; this is true since any binary submodular function is the rank function of some finite matroid \cite{oxley2011matroid}.  

In our data, student scheduling constraints largely retain submodularity. This is because the vast majority of classes are scheduled at a set of predetermined times, e.g., Tues/Thu at 1pm-2:15pm. Thus, if class $A$ conflicts with $B$ and $B$ conflicts with $C$, then $A$ conflicts with $C$, avoiding issues similar to those described above.   

\subsection{Justice Criteria}
\label{sec:justice_criteria}
%We assess the performance of the allocation algorithms by comparing the output allocation $A$ in terms of different criteria. \paula{(to self) consider rephrasing this, not going into the allocation algorithms just yet, but just comparing allocations}. 
We evaluate the performance of course allocation mechanisms according to \emph{fairness} and \emph{social welfare} criteria. 
We consider three efficiency criteria. 
The \emph{Utilitarian Social Welfare} (USW) sums the total welfare of agents:
$\USW(X)=\frac{1}{n}\sum_{i\in N}v_i(X_i)$.
The USW criterion does not account for potential welfare disparities. For example, from the USW perspective, a course assignment that offers six courses to one student and none to another is equivalent to one that assigns three courses each. 
% To evaluate a more equitable course distribution, we utilize \emph{egalitarian} objectives. 
% The \emph{leximin} welfare criterion maximizes the welfare of the least well-off agent; subject to that, it maximizes the welfare of the second least well-off agent and so on. More formally, given an allocation $X$, let $\vec u(X)$ be the vector of agent utilities sorted in increasing order of utilities. An allocation $X$ maximizes the leximin objective if its sorted utility vector lexicographically dominates that of any other allocation. 

% A leximin allocation heavily skews towards lower utility agents. 
% For example, it prefers an allocation that offers all agents a utility of $2$ to one that offers $n-1$ agents a utility of $5$, and one agent a utility of $1$. 
The \emph{Nash welfare} \citep[Chapter 3]{moulin2003fairdivision} strikes a balance between utilitarian and egalitarian approaches.
An allocation that maximizes the Nash welfare first minimizes the number of agents with zero utility. Subject to that, it maximizes the \emph{product} of agent utilities, i.e. the value $\NSW(X) = \prod_{i \in N_{>0}(X)} v_i(X_i)$, where $N_{>0}(X)$ is the set of agents with positive utility under $X$. 
% \item Envy freeness violations: 
% $$ENVY(A)=|\{i\in N|\exists k \in N: v_i(\vec a_i)< v_i(\vec a_k)\}|$$
% \item Envy freeness up to one item violations: 
% $$ENVY\_1(A)=|\{i\in N|\exists k \in N:\forall j \in M:   a_k^j\neq 0\rightarrow v_i(\vec a_i)< v_i(\vec a_k -\mathbbm{1}_j)\}|$$

We also evaluate allocations based on \emph{fairness} metrics.
Given an allocation $X$, we say that agent $i$ \emph{envies} agent $j$ \citep{foley1967envy} if $v_i(X_i) < v_i(X_j)$. An allocation is \emph{envy-free} (EF) if no agent envies another. 
Envy-free allocations are not guaranteed to exist (consider a setting with two agents and one single item); this gave rise to a the canonical notion of \emph{envy-freeness up to one item} (EF-1) \citep{budish2011ef1,lipton2004ef1} and \emph{envy-freeness up to any item} (EF-X) \citep{caragiannis2019unreasonable,plaut2020efx}. 
An allocation $X$ is EF-1 if for any two agents $i$ and $j$, if $i$ envies $j$, then there exists some item $g\in X_j$ such that $v_i(X_i) \ge v_i(X_j-\Ind_g)$: removing some item from agent $j$'s bundle eliminates agent $i$'s envy. 
An allocation is EF-X if removing \emph{any} item from $j$'s bundle eliminates agent $i$'s envy towards $j$. 

We also wish to examine \emph{share-based} fairness concepts, e.g., the maximin share \cite{budish2011ef1}. 
The maximin share guarantee of agent $i$ is computed by letting agent $i$ partition the items into $n$ bundles, and receiving the worst one. 
More formally, $\MMS_i  = \max_{X}\min_j v_i(X_j)$. 
While the maximin share is a well-established and well-studied share-based concept (see, e.g.,\cite{amanatidis2017approximation,barman2019fair,barman2020mms,budish2011ef1,caragiannis2019unreasonable,garg2020mms,kurokawa2018fairenough}), it is inappropriate in our setting. 
The reason is simple: most students rate the vast majority of classes at $1$, i.e., they do not want to take them. 
Thus, any partition of item copies to $n$ bundles will result in at least one bundle whose value is $0$. 
This implies that the maximin share of the vast majority of students is $0$. 
Since we are interested in analyzing some share-based fairness criterion, we use the \emph{pairwise maximin share} \cite{caragiannis2019unreasonable}. 
An allocation $X$ is \emph{Pairwise Maximin Share Fair} (PMMS-fair) if for any pair of agents $i$ and $j$, agent $i$'s valuation of their own bundle is at least the valuation of the worst bundle they could get in a two way division of $X_i \cup X_j$.

\subsection{Allocation Mechanisms}\label{sec:allocation_algs}
An allocation algorithm takes as input a set $G$ of $m$ courses with $\vec q =(q(g_1),...,q(g_m))$ seats per course, and a set $N$ of $n$ students with valuations $(v_1,\dots,v_n)$. 
Its output is some feasible allocation $X$ of the course seats. 
We consider four different allocation algorithms: Serial Dictatorship, Round Robin, an Integer Linear Program, and Yankee Swap.    

\begin{description}[leftmargin=0cm]
\item[Serial Dictatorship \cite{hatfield2009serialdictatorship}:] the course allocation algorithm used by \UMass for student enrollment uses a variant of Serial Dictatorship (SD). 
The allocation system opens to students in order of seniority. PhD students enter first, followed by MS students; next, the system opens for undergraduate enrollment in decreasing order of seniority. 
The first student to access the platform enrolls in all their desired courses. 
Subsequently, the following students enroll in all desired courses that have available seats left. 
We emulate this process by letting students pick in order of their academic status, and in random order within members of the same academic cohort.
% More formally, the algorithm works as follows. Initially, students start with an empty bundle and all course seats are available. The first student picks a bundle of classes that maximizes their utility and course capacities are updated accordingly.  
% The second student does the same, considering these updated capacities. 
% The algorithm terminates once all students picked classes. 
We assume that students pick clean bundles, i.e., that students do not enroll in classes that conflict with one another, or that they do not want to take.
This assumption is reasonable: the actual course enrollment system at \UMass does not allow students to pick classes with conflicting schedules or enroll in more classes than their current status allows (usually no more than six classes, or 19 academic credits per semester). %In addition, students may not enroll in classes for which they do not satisfy the prerequisite requirement. 
Students can enroll in classes they do not want to take; however, since SD is strategyproof, students have no incentive to do so. 
We do mention that students are often uncertain about what classes they actually want to take, and tend to enroll in more classes than they intend on keeping (subsequently dropping out during the grace period early in the semester). 
Thus, assuming that students only select clean bundles makes our version of SD appear more fair than it potentially is in practice: students who access the platform early get a wide range of available classes and `hoard' classes; those who log in later face a considerable disadvantage in securing desired classes. These issues with SD are known from prior works on course allocation, e.g., \citet{budish2011ef1,budish2012roundrobincourseallocation}. 
\item[Round Robin:] the \textit{Round Robin algorithm} (RR) \citep{brandt2016handbookfairdiv} (also known as the draft mechanism \cite{budish2012roundrobincourseallocation}), operates in rounds. In each round, students are assigned one available seat that they like, i.e., one that offers them a marginal gain of $1$. This continues until no such seats are available. The order in which students select their seats is fixed. We again prioritize students according to academic status.
\item[Integer Linear Program:] allocation by Integer Linear Program (ILP) finds an allocation $X$ that maximizes the utilitarian social welfare subject to course capacity constraints, and linear inequality constraints that encode student preferences (see \Cref{subsec:valuation_functions}). Since both the objective and constraints are linear, the ILP returns an optimal integer allocation with respect to \USW. 
% In particular we solve the following ILP.

% \smallskip
% \noindent \textbf{Maximize:}
% \begin{equation}
%     \USW(A)=\frac{1}{n}\sum_{i\in N}v_i(A_i).
% \end{equation}
% \noindent \textbf{Subject to:}
% \begin{equation}
%     \left[ \begin{array}{cccc}
%         Z^1&&& \\
%         &Z^2&& \\
%         &&\ddots& \\
%         &&&Z^n \\
%     \end{array} \right] 
%     \left[ \begin{array}{c}
%         A_1 \\
%         \vdots \\
%         A_n \\
%     \end{array} \right] 
%     \leq
%     \left[ \begin{array}{c}
%         \Vec{b}^1 \\
%         \vdots \\
%         \Vec{b}^n \\
%     \end{array} \right] .
% \end{equation}
\item[Yankee Swap:] in the \textit{Yankee Swap algorithm} (YS) \citep{viswanathan2023yankeeswap,viswanathan2023general} agents sequentially pick items. 
In each round, YS selects the lowest utility student and lets them pick a seat that offers them a marginal benefit of $1$. If no such seat is available, they may steal such a seat from another student, as long as that student can recover their utility by either taking an unassigned seat or stealing a seat from yet another student. This goes on until the last student takes an unassigned seat. If no such \emph{transfer path} exists, the student is not elected again. 
The algorithm terminates when no student can further benefit from enrolling in courses with available seats, or when students are solely interested in stealing seats from students who cannot recover their utility. 
When agents have binary submodular valuations, YS offers several theoretical guarantees: it outputs a leximin, EF-X allocation, which also maximizes \USW.
In addition, YS is \emph{truthful}: no agent can increase their utility by misreporting their preferences, e.g., falsely state that they want to enroll in an undesirable class, or say that they do not want to enroll in a desirable class. 

The main computational hurdle in implementing YS is computing the \emph{item exchange graph} (see \Cref{sec:yswithmult} for more details): this is a graph over item types, where there is a directed edge from item $g$ to item $g'$ if the agent who owns item $g$ can retain their current utility by exchanging $g$ for $g'$. This graph is used to compute the transfer paths utilized in YS, and is the most computationally expensive aspect of its runtime. \Cref{sec:yswithmult} describes how we exploit the underlying problem structure in order to derive additional improvements to the YS. 

Like serial dictatorship and Round Robin, Yankee Swap employs a tie breaking scheme to determine which student gets to pick/steal a seat if there is more than one student with the lowest utility. We similarly break ties in favor of students of higher academic status, and use a deterministic (randomly selected) tie-breaking scheme within cohorts.  
\end{description}

\section{Adapting the Yankee Swap Mechanism in the Course Allocation Domain}
\label{sec:yswithmult}
% The original Yankee Swap algorithm \cite{viswanathan2023yankeeswap} assumes no item multiplicity. When items have multiple identical copies, considering each item copy individually becomes highly inefficient, as the problem size scales by the number of course \emph{seats} rather than the number of courses, which is typically significantly smaller. 
% We propose a modified version of the YS that allows for item multiplicity. 

%In the original problem formulation, where each item has a single copy, 
The main computational overhead of the Yankee Swap algorithm \cite{viswanathan2023yankeeswap} stems from maintaining the \emph{exchange graph} $\mathcal{G}$. 
Given an allocation $X$, $\mathcal{G}(X)$ is a directed graph over the set of items $G$. 
There is an edge from an item $g\in X_i$ to another item $g'\in G$ if $v_i(X_i)=v_i(X_i -\mathbbm{1}_g +\mathbbm{1}_{g'}) $. In other words, there is an edge from $g$ to $g'$ if the agent who owns $g$ under $X$ can replace it with $g'$ without reducing their utility. 
If no agent owns item $g\in G$ (i.e., $g\in X_0$), then $g$ has no outgoing edges.

The Yankee Swap algorithm works as follows. Initially, all agents are assigned empty bundles and all items are unassigned, i.e., the initial allocation $X$ is such that $X_{0,g}=1$ and $X_{i,g}=0$ for all item types $g\in G$ and agents $i\in N$. 
At each round, we pick an agent $i$ with the lowest utility, breaking ties arbitrarily. 
Let $F_i(X)=\{g\in G: \Delta_i(X_i,g)=1\}$ be the set of items that offer agent $i$ a marginal benefit of $1$ given their current bundle. 
% We introduce dummy nodes $i$ and $t$ to the exchange graph. We add a directed edge from $i$ to every item in $F_i(X)$, and a directed edge from each item $g \in X_0$ to $t$. 
We search for a shortest path from $F_i(X)$ to items in $X_0$ in the resulting graph and execute it. Note the execution of  the path $(g_0,g_1,\dots,g_p)$ corresponds to agent $i$ taking item $g_0$, and the agent who owns item $g_0$ replaces it with item $g_{1}$, and so on, until we reach the unassigned item $g_p$, which is now assigned to the last agent in the path. This is referred to as \emph{path augmentation}. 
If no such path exists, agent $i$ is kicked out of the game. 
We repeat this process until there are no agents left playing or all items have been allocated. 

The original problem formulation assumes no item multiplicity. When items have multiple identical copies, considering each item copy individually becomes highly inefficient, as the problem size scales by the number of course \emph{seats} rather than the number of courses, which is typically significantly smaller. 

\subsection{Accounting for Item multiplicity}
\label{subsec:item_mult}
We propose a modified version of the Yankee Swap algorithm that allows items to have multiple identical copies, i.e., $q(g)> 1$. 
Since items have multiple copies, they might have multiple owners, which requires a redefinition of the exchange graph. The exchange graph $\mathcal{G}$ is now a directed graph over the set of item types $G$, in which there is an edge of the form $(g, g')$ if there \emph{exists} some agent $i$ who owns a copy of $g$, and is willing to exchange it for a copy of $g'$. 

Since items can have multiple owners, an edge $(g, g')$ in the exchange graph may result from one or more agents willing to make the exchange, thus, we need to maintain a list of agents willing to make the exchange. 
We define a \emph{responsible agents tuple} $R$ consisting of $m^2$ lists of agents that keep track of the agents responsible for each of the edges. Here, $R_{g\to g'}$ is the set of agents who own a copy of $g\in G$ and are willing to exchange it for a copy of $g'\in G$:
$$R_{g\to g'}=\{i \in N: X_{i,g}>0 \wedge v_i(X_i)=v_i(X_i -\mathbbm{1}_{g}+\mathbbm{1}_{g'}) \}.$$ 
$R_{g\to g'}=\emptyset$ if and only if there is no edge $(g, g')$  in $\mathcal{G}$: no agent who owns a copy of $g$ wants to exchange it for a copy of $g'$. 
Initially, $R_{g \to g'}=\emptyset$ for every $g, g' \in G$ since no items have been allocated. 
Once a path transfer is executed, every agent who had their bundle changed in the process might change the $R_{g\to g'}$ values. 
If done naively, these checks can be time-consuming: theoretically one needs to check whether any item in agent $i$'s bundle can be exchanged for any other item. 
To avoid this overhead, we observe that the set of classes that agent $i$ can potentially want to exchange their classes for is a subset of the set of classes they approve.
Let $D_i=\{g\in G| v_i(\mathbbm{1}_g)=1\}$ be the set of agent $i$'s approved classes. 
If agents have binary submodular valuations over the items, then for any allocation $X$, $F_i(X)\subseteq D_i$. In particular, $F_i(X)=D_i$ when $X$ is the initial empty allocation. 

Our implementation of YS (see \Cref{alg:ys}) incorporates key adjustments to account for item multiplicity and efficiency. 
Initially all items copies are unassigned, i.e., $X_{0,g}=q(g)$ for all $g\in G$. 
Second, and more importantly, we run path augmentations over item copies while keeping track of item owners. 
Third, instead of reconstructing the exchange graph after each iteration, we incrementally update it, resulting in a more efficient execution.

Path augmentation occurs in two distinct phases. In the first phase, we identify the agents involved in the transfer path and update the allocation matrix $X$. For instance, if agent $j$ exchanges item $g$ for $g'$, we update the matrix to reflect that the agent now possesses $g'$ and no longer has $g$. 
Next, we update the exchange graph $\mathcal{G}(X)$ and the tuple $R$ to account for items that agents along the path \emph{lose}. 
Specifically, if agent $j$ loses the item $g$ then for any item $h$ approved by agent $j$, $h \in D_j$, we remove agent $j$ from $R_{g \to h}$. 
If $R_{g \to h}$ becomes empty, we remove the edge $(g, h)$ from $\mathcal{G}$. 
In the second phase, we update the exchange graph and the matrix $R$ based on items that agents along the path \emph{receive}. 
For each agent $j$ identified in the previous phase, we check their updated bundle and determine whether they are willing to exchange any item in their bundle $X_j$ for items in their desired set $D_j$. We then adjust the tuple $R$ and update the edges in the exchange graph accordingly. 
Since $D_j$ is the set of classes that each student intrinsically approves, it remains invariant throughout the run of Yankee Swap, and is relatively easy to maintain.
% We use $D_j$ instead of $F_j(X)$ for two reasons. 
% First, $D_j$ can be computed beforehand, whereas $F_j(X)$ must be recomputed after every allocation change. More importantly, $D_j$ captures potential exchanges that the agent was previously willing to make but may no longer be interested in after the allocation update, thus, using $D_j$ ensures that we accurately remove any edges corresponding to outdated exchange preferences that $F_j(X)$ might overlook.

\begin{algorithm}
\caption{Yankee Swap with Item Multiplicity}\label{alg:ys}
\begin{algorithmic}[1]
\Require{A set of item types $G$ with $\vec q =(q(g_1),...,q(g_m))$ copies each, and a set of agents $N$ with binary submodular valuations $\{v_i\}_{i\in N}$ over the items}
\Ensure{A clean leximin allocation}
\State $X = (X_0,X_1,...,X_n) \gets (\vec q, \vec 0, ..., \vec 0)$ 
\State $\R_{g\to g'}\gets \emptyset$ for all $g\in G, g' \in G$
\State Build graph $\mathcal{G}$ with one node per item type $g\in G$
\State $U \gets N$
\While{$U \neq \emptyset$ } 
 \State  Let $i\in \arg\min_{j\in U} v_j(X_j) $ 
\State Check if there is a path $P=(g_0,...,g_{p})$ where $g_0\in F_i(X)$ and $g_{p}\in X_0$ 
\If{a path exists} 
\State{\texttt{//Path Augmentation Phase 1}}
\State $I \gets \{i\}$, $X_i\gets X_i +\mathbbm{1}_{g_0}$
\For{each pair $(g, g')$ in $P$} 
% \State Find agent $k_{l+1}$ willing to exchange item of type $g_{j_l}$ for item of type $g_{j_{l+1}}$  
\State Let $j \in R_{g\to g'}$
\State $I \gets I\cup\{j\}$, $X_{j}\gets X_j -\mathbbm{1}_{g}+\mathbbm{1}_{g'}$
\For{each item type $h \in D_{j}$  } 
\If {$j\in R_{g\to  {h}}$ and $X_{j,g}=0$} 
\State $R_{g\to {h}}\gets R_{g\to {h}} \setminus \{j\}$
\If{$R_{g\to h}=\emptyset$}
\State Remove edge $(g, h)$ from $\mathcal{G}$
\EndIf
\EndIf
\EndFor
\EndFor
\State $X_0\gets X_{0} -\mathbbm{1}_{g_p}$
\State{\texttt{//Path Augmentation Phase 2}}
\For{each agent $j \in I$ }
\For{each item type $h \in X_{j}$ and item type $h' \in D_{j}$ 
} 
\If{$j\in R_{h\to h'}$}  
\If{$v_j(X_j)> v_j(X_j-\mathbbm{1}_{h}+\mathbbm{1}_{h'})$} 
\State $R_{{h}\to {h'}}\gets R_{{h}\to {h'}} \setminus \{j\}$
\If{$R_{{h}\to {h'}}=\emptyset $}
\State Remove edge $(h, h')$ from $\mathcal{G}$
\EndIf
\EndIf
\Else{}
\If{$v_i(X_i)\leq v_i(X_i-\mathbbm{1}_{h}+\mathbbm{1}_{h'})$} 
\State $R_{{h}\to {h'}}\gets R_{{h}\to {h'}} \cup \{j\}$ 
\If{$|R_{h\to h'}|=1$}
\State Add edge $(h, h')$ to $\mathcal{G}$
\EndIf
\EndIf
\EndIf
\EndFor
\EndFor
\Else{} 
    \State $U \gets U \setminus \{i\}$ 
    % \State $\xi_i\gets -\infty$ 
\EndIf
\EndWhile
\State \Return $X$
\end{algorithmic}
\end{algorithm}
% \Cref{alg:ys} is a modified version of the Yankee Swap algorithm \citep{viswanathan2023yankeeswap} that allows item multiplicity. 
% Note that the algorithm keeps the same structure as the original algorithm. However, two subroutines, \textit{update allocation} and \textit{update exchange graph} are modified (see Algorithms \ref{alg:allocation} and \ref{alg:graph} respectively). 

\subsection{Analysis}
\label{subsec:runtime}
Let $q_\text{total}=\sum_{g\in G}q(g)$ be the total number of item copies. 
A naive implementation of Yankee Swap searches for a shortest path in the exchange graph in time quadratic in $q_{\text{total}}$, which is significantly larger than the number of item types $m$. We reduce the runtime dependency to be on $m$ when searching for transfer paths. 
Second, it turns out that calls to student valuations can be expensive; thus, our goal is to reduce the number of student valuation calls. 

We assume that each agent is limited to a maximum clean bundle size of $c_{\max}$, i.e., $v_i(X_i)\leq c_{\max}$ for all $i\in N$ and any possible bundle $X_i$. 
Similarly, we assume that each agent can get a positive marginal utility of at most $d_{\max}$ items: $|D_i|\leq d_{\max}$ for any $i\in N$. 
Finally, we define $p_{\max}$ as the maximum length of any transfer path. 

We now analyze the time complexity of Algorithm \ref{alg:ys}. Since the allocation $X$ is a matrix, updating an agent's bundle can be done in $O(1)$ time. 
In each iteration, the algorithm either allocates an unassigned item or removes an agent from the game. 
Thus, the outer loop of the algorithm runs at most $q_\text{total}+n$ times. 
Identifying the unhappiest agent can be done in $O(\log n)$ time and, given that there are $O(m)$ nodes in the exchange graph $\mathcal{G}(X)$, we can find the shortest path in $O(m^2)$ time. 
Path augmentation is the most expensive part of the algorithm. Let $q_{\max}=\max_{g\in G} q(g)$ be the largest number of copies of any individual item. 
Since $R_{g\to g'}$ can have at most $q(g)$ elements, determining whether an agent is willing to exchange $g$ for $g'$ takes $O(\log q_{\max})$ time. 
Hence, phase 1 of the path augmentation can be computed in $O(p_{\max} d_{\max} \log q_{\max})$ time, and phase 2 in $O(p_{\max}d_{\max}c_{\max}(\log q_{\max}+\tau))$ time. 
Combining these observations yields following result. 
\begin{theorem}
    Algorithm \ref{alg:ys} runs in $O\left((q_\text{total}+n)(\log n + m^2 + p_{\max}d_{\max}c_{\max}(\log q_{\max}+\tau)\right)$
\end{theorem}
In the context of course allocation, $c_{\max}=6$ is reasonable, since this is the maximal number of classes students are allowed to take in a single semester at \UMass. 
Students typically assign a high score to significantly fewer classes than $m$, thus $d_{\max}$ is also $\ll m$. 
While transfer paths could potentially be of length $m$, they are far shorter in practice, often consisting of no more than two or three item swaps (see \Cref{appendix:transfer_paths}). 
Finally, separating the exchange graph representation from the agents through the tuple $R$ considerably reduces the number of student valuation function calls.
\section{Implementation, Experiments and Results}
\label{sec:experiments}
We implement a general suite of tools for the fair allocation of indivisible items. 
We allow for item multiplicities and allow general constraint-based agent valuations. 
Our implementation includes algorithms for the mechanisms discussed in \Cref{sec:allocation_algs}: Serial Dictatorship (SD), Round Robin (RR), Yankee Swap (YS), and an Integer Linear Program (ILP). 
% Each algorithm receives a list of item and agent objects, and returns allocations as a binary allocation matrix, where each entry indicates whether a given item copy is assigned to a particular agent. 
In the context of course allocation, we model items as courses with features such as catalog number, section, and meeting time, and agents as students with binary valuation functions. 
The full set of courses, referred to as the \emph{schedule}, is derived from the Fall 2024 course offerings of the Computer Science department at \UMass. It includes 96 courses, each with real section details, enrollment capacities, and meeting times. 

We implement the justice criteria described in \Cref{sec:justice_criteria}, to evaluate the allocations. 

We describe the simulation of real and synthetic students, the practical implementation of the justice criteria, followed by the execution of allocation algorithms for two scenarios. First, we run the algorithms on a scaled instance using only real student responses. Next, we consider a full-scale instance, combining real respondents with synthetic students and using real course capacities. 
Finally, we discuss the results and key findings from these experiments.
\subsection{Simulated Students}
We model \emph{real students} based on survey responses and generate \emph{synthetic students} using the method outlined in \Cref{sec:synth_students}. 
From the survey responses regarding the maximum number of courses students wish to enroll in, along with corresponding simulated values, we define the maximum enrollment capacity for each student. We take the minimum between this value and the capacity allowed by the department, which is 6 for undergraduate students and 4 for gradute students. 

Both real and synthetic students express preferences for each course as numerical values ranging from 1 to 8. However, since the allocation algorithms require binary preferences, we map these cardinal preferences to approval preferences by constructing a set of preferred courses using a top-$k$ approach: for each student, we determine a preference threshold based on their $k$-th preference and the corresponding set of preferred courses. 
\begin{example}
Suppose that students approve of their top $2$ ranked classes, i.e., $k =2$. 
The translation from survey responses to approval votes is as follows: 
\begin{center}
    \begin{tabularx}{0.4\textwidth}{cXXXX}
\toprule
    class: & A & B &C &D\\
    \toprule
    Student 1: &6 & 3 &2 & 3\\
    \midrule
    Student 2: & 3 & 1 & 1 & 1\\
    \midrule
    Student 3: & 5 & 5 & 5 & 5 \\
    \bottomrule
    \end{tabularx}
    {\Huge{$\Rightarrow$}}
        \begin{tabularx}{0.4\textwidth}{cXXXX}
\toprule
    class: & A & B &C &D\\
    \toprule
    Student 1: &1 & 1 &0 & 1\\
    \midrule
    Student 2: & 1 & 0 & 0 & 0\\
    \midrule
    Student 3: & 1 & 1 & 1 & 1\\
    \bottomrule
    \end{tabularx}
\end{center}
Since $3$ is the 2nd highest value reported by student 1, classes $A$ $B$ and $D$ are approved.
Student 2 only ranks one class at a value greater than $1$, so that class is the only one approved. 
Student 3 ranked all classes at the same score, and rated them above 1. Thus, they approve of all classes. 
\end{example}
Translating students' numerical scores to approval votes via a top-$k$ approach eliminates any inherent bias that students might exhibit in reporting their scores: being more enthusiastic about classes does not offer an advantage. In addition, it avoids breaking ties in an arbitrary fashion: all classes that receive the same rating receive the same approval score.  
\begin{figure}[t]
        \centering        
        \includegraphics[width=\textwidth]{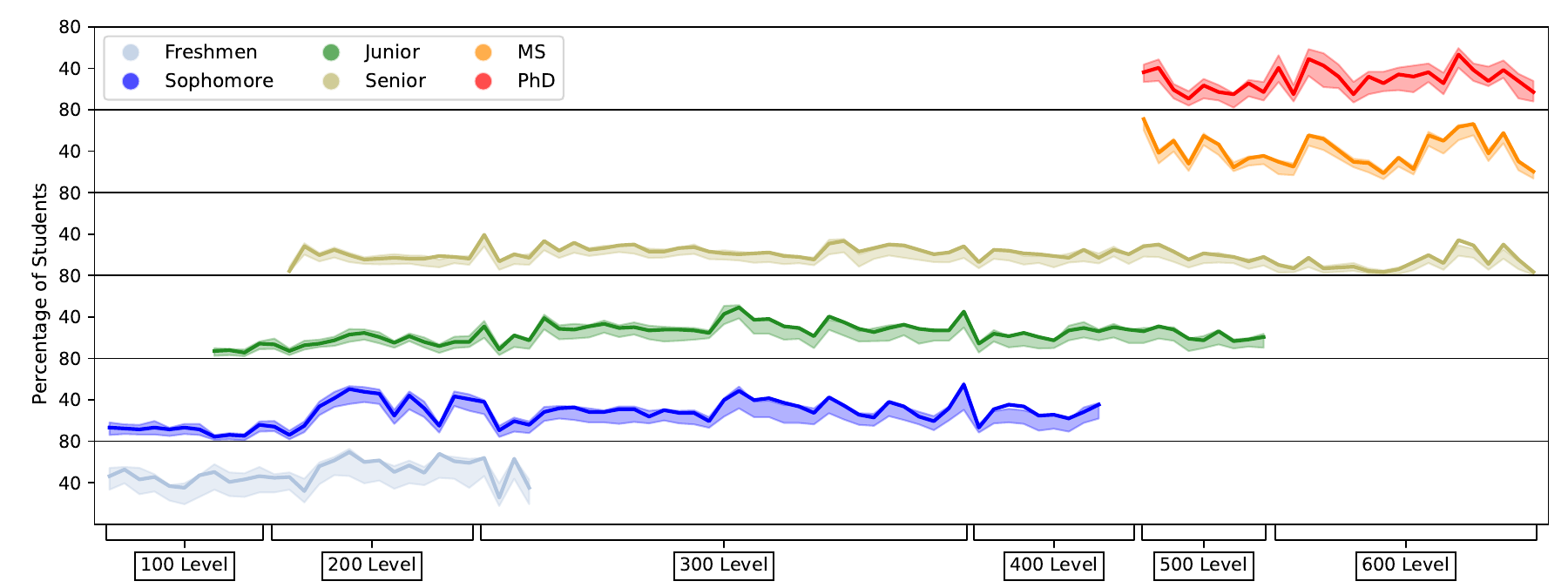}
        \caption{Percentage of students who include each course in their set of preferred courses, grouped by student status. Thick colored lines represent real student data, while the shaded regions indicate the range of percentages observed across 100 synthetic student simulations.}
        \label{fig:preferences}
        \Description{Results on larger student cohort.}
\end{figure}
Synthetic students and real students exhibit similar approval preferences. 
\Cref{fig:preferences} lists the percentage of real and synthetic students, grouped by status, who include each course in their set of preferred courses, considering $k=10$. 
Real students are given by all real responses, and we augment them with synthetic students until reaching the original quantities per status shown in \Cref{table:student_quantities}. The figure considers 100 different seeds for randomly generating the synthetic students.
While synthetic students tend to slightly underestimate preferences, they follow similar percentage trends as real students, providing a reasonable approximation for course demand.

% Using this information --the maximum enrollment capacities and the mapped preferences-- we construct the constraint matrices that represent valuation functions defined by linear constraints (see Section \ref{subsec:valuation_functions}). These matrices are then used to generate the corresponding student objects.

\subsection{Evaluation metrics}
We evaluate how each of the benchmarks performs according to the justice criteria described in \Cref{sec:justice_criteria}.
Given an allocation $X = (X_0;X_1,\dots,X_n)$, we compute the following metrics.
\begin{description}[leftmargin = 0cm]
\item[Welfare Metrics:] we compute the utilitarian welfare $\USW(X)$ and the Nash welfare $\NSW(X)$. We present utilitarian welfare as the percentage of total seats assigned, while Nash welfare is normalized by taking the $n$-th root, i.e., $\sqrt[n]{\prod_{i\in N_{>0}(X)}v_i(X_i)}$, since the true Nash welfare values are galactic numbers. 
Since the Nash welfare takes the geometric average of the positive student utilities, we complete the picture by reporting the number of students who receive an empty bundle, i.e., students who ended up receiving no classes that they like. 
\item[Envy Metrics:] we count the number of times an agent envies another, i.e., 
$$\ENVY(X) = \sum_{i = 1}^n\sum_{j = 1}^n \mathbbm{1}[v_i(X_i)<v_i(X_j)].$$
In addition, we count the number of times an EF-1 violation occurs, i.e., the number of times that an agent envies another even after any item is removed from the envied agent's bundle. 
$$\ENVY\text{-}1(X) = \sum_{i = 1}^n\sum_{j = 1}^n \mathbbm{1}[\forall g \in X_j:v_i(X_i)<v_i(X_j-\mathbbm{1}_g)].$$
\item[Pairwise Maximin Share:] as discussed in \Cref{sec:justice_criteria}, the maximin share guarantee is an inappropriate measure in the course allocation domain, as most students have an MMS value of $0$. Thus, we count the number of times that the PMMS guarantee is violated under the allocation $X$: 
\newcommand{\PMMS}{\texttt{PMMS}\xspace}
$$\PMMS(X) = \sum_{i= 1}^n \sum_{j = 1}^n \mathbbm{1}[v_i(X_i) < \max_{T \preceq X_i + X_j}\min\{v_i(T),v_i(X_i + X_j - T)]$$
\end{description}
The algorithms we use offer some justice guarantees. The ILP maximizes utilitarian welfare by definition, the Yankee Swap framework outputs a leximin allocation \cite{viswanathan2023yankeeswap}, which is guaranteed to maximize utilitarian and Nash welfare \cite{babaioff2021EF} when agents have binary submodular utilities. 
In addition, Yankee Swap outputs EF-X allocations \cite{caragiannis2019unreasonable,plaut2020efx}: if agent $i$ envies agent $j$, then the removal of any item from $j$'s bundle eliminates envy. These allocations are also guaranteed to satisfy PMMS when agents have binary submodular utilities. The Round Robin mechanism is EF-1 when agents have additive utilities \cite{brandt2016handbookfairdiv}, but this guarantee does not necessarily hold when agents have non-additive utilities; as a result, the Round Robin mechanism does exhibit some EF-1 violations in our evaluation. 

\subsection{Experiments}
\begin{figure}[t]
        \centering  
        \includegraphics[width=\textwidth]{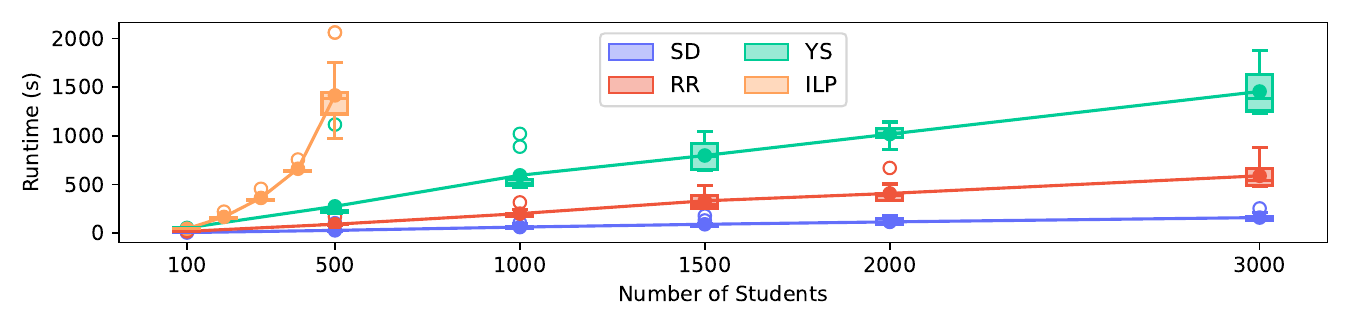}    
        \caption{Runtime of the four allocation algorithms as a function of the number of students. Boxplots represent the distribution of runtime values across 10 runs with different random seeds for each instance.}
        \label{fig:runtime}
        \Description{}
\end{figure}
We run our experiments on a server with 8 core Intel(R) Xeon(R) Gold 5222 CPU @ 3.80GHz and 1.5 TB RAM.
We first examine the scalability of each of our benchmarks. 
For each cohort size, we use either only real students responses, or a mixture of real and synthetic students to reach the target cohort size. 
To ensure each instance remains realistic, we adjust the student composition to match the proportional distribution of each academic status in the department as shown in \Cref{table:student_quantities}, and we scale course capacities accordingly.
Once students are generated, all sequential algorithms (SD, RR and YS) utilize the same student tie-breaking scheme: students with a higher academic status are given higher priority, and are ordered uniformly at random within each group. 

\Cref{fig:runtime} shows that sequential approaches such as SD, YS and RR scale well as the number of students grows, whereas the ILP suffers from exponential blowup in its runtime, making it infeasible to evaluate for instances of more than 500 students. 

To compare all four benchmarks, we utilize a reduced instance consisting only of students generated from a subset of real survey responses. 
%Student response rates vary across different statuses, but many course offerings are not restricted to a specific academic status --- there is substantial overlap in course preferences among students from different groups. 
According to \Cref{table:student_quantities}, seniors have the lowest response rate of 20.42\%. We scale each course capacity to this value, and sample a subset of real responses from each other status to approximately match this rate, obtaining an instance of 471 students. 
We subsample 100 different instances by varying which students are included in the sample. 
We compare \textcolor{blue}{SD (in blue)}, \textcolor{red}{RR (in red)}, \textcolor{blue!30!green}{YS (in green)} and \textcolor{orange}{ILP (in orange)}.
\begin{figure}[t]
        \centering   
        \begin{subfigure}[b]{0.25\textwidth}
        \centering
        \includegraphics[width =\textwidth]{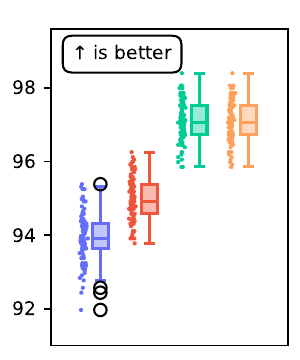}  
        \caption{\% of seats assigned}\label{fig:num-seats-reduced}
        \end{subfigure}
        \begin{subfigure}[b]{0.25\textwidth}
        \centering
        \includegraphics[width = \textwidth]{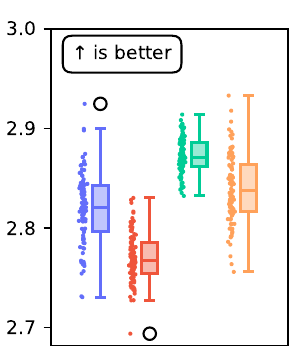}
        \caption{Nash welfare}\label{fig:nash-welfare-reduced}
        \end{subfigure}
        \begin{subfigure}[b]{0.25\textwidth}
        \centering
        \includegraphics[width = \textwidth]{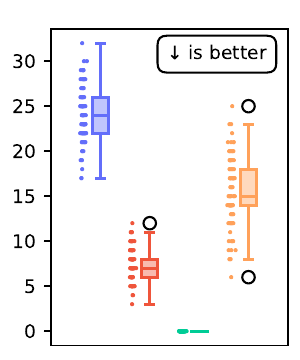}
        \caption{Num. of $0$ util. agents}\label{fig:zeros-reduced}
        \end{subfigure}
        \caption{   
        Welfare metrics for the reduced student instance. \Cref{fig:num-seats-reduced} compares \% of allocated seats (USW), \Cref{fig:nash-welfare-reduced} compares the Nash welfare, and \Cref{fig:zeros-reduced} presents the number of empty bundles.}
        \label{fig:reduced_welfare_metrics}
        \Description{Welfare metrics for reduced students}
\end{figure}

\Cref{fig:reduced_welfare_metrics} summarizes the relative welfare performance of the four algorithmic benchmarks for reduced student instances. Since YS and ILP maximize welfare, they assign the maximal number of seats, whereas RR and SD fail to maximize welfare. YS is theoretically guaranteed to outperform all other benchmarks on Nash welfare (\Cref{fig:nash-welfare-reduced}); the fact that other benchmarks sometimes outperform it is explained by \Cref{fig:zeros-reduced}: YS first minimizes the number of zero utility agents (all students were assigned at least one desirable class in every run of Yankee Swap), and then maximizes the Nash welfare. On the other hand, every other benchmark leaves at least some students with zero utility; SD exhibits the worst performance, leaving 23 students (roughly 5\% of the cohort) with zero utility on average.  
\begin{figure}[t]
        \centering   
        \begin{subfigure}[b]{0.25\textwidth}
        \centering
        \includegraphics[width = \textwidth]{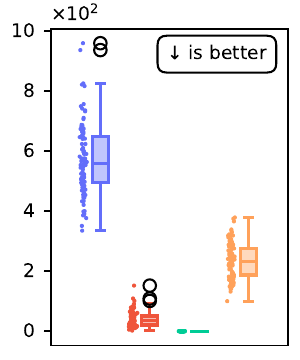} 
        \caption{PMMS violations}\label{fig:PMMS-reduced}
        \end{subfigure}
        \begin{subfigure}[b]{0.25\textwidth}
        \centering
        \includegraphics[width = \textwidth]{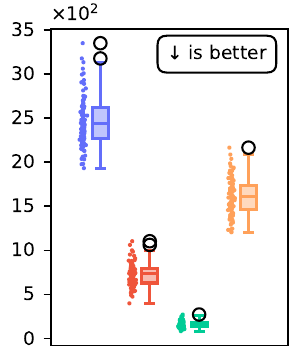}
        \caption{Envy violations}\label{fig:EF-reduced}
        \end{subfigure}
        \begin{subfigure}[b]{0.25\textwidth}
        \centering
        \includegraphics[width = \textwidth]{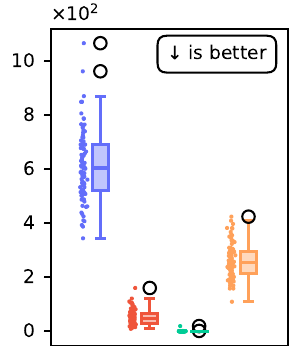}
        \caption{EF-1 violations}\label{fig:EF1-reduced}
        \end{subfigure}
        \caption{   
        Fairness metrics for the reduced student instance. \Cref{fig:PMMS-reduced} shows the number of PMMS violations, \Cref{fig:EF-reduced} shows the number of envious agents, and \Cref{fig:EF1-reduced} shows the number of EF-1 violations.}
        \label{fig:reduced_fairness_metrics}
        \Description{Fairness metrics for reduced students}
\end{figure}
\Cref{fig:reduced_fairness_metrics} summarizes how the four benchmarks fare in terms of fairness. 
Since YS is guaranteed to be EF-1 and offers each agent their PMMS share, it significantly outperforms the other methods in terms of both PMMS and EF-1 violations. In addition, it exhibits significantly fewer envy violations than the other methods.
%this is because it satisfies a property called \emph{downwards envy-freeness}: agents never envy lower priority agents, i.e. those who are lower than them in the priority tie-breaking scheme employed by YS \cite{babaioff2021EF}. 
Perhaps unsurprisingly, the allocations output by the SD mechanism highly favor students who access the system early, resulting in significant envy and PMMS violations. Finally, we note that  
\begin{figure}
    \centering
\includegraphics[width=\textwidth]{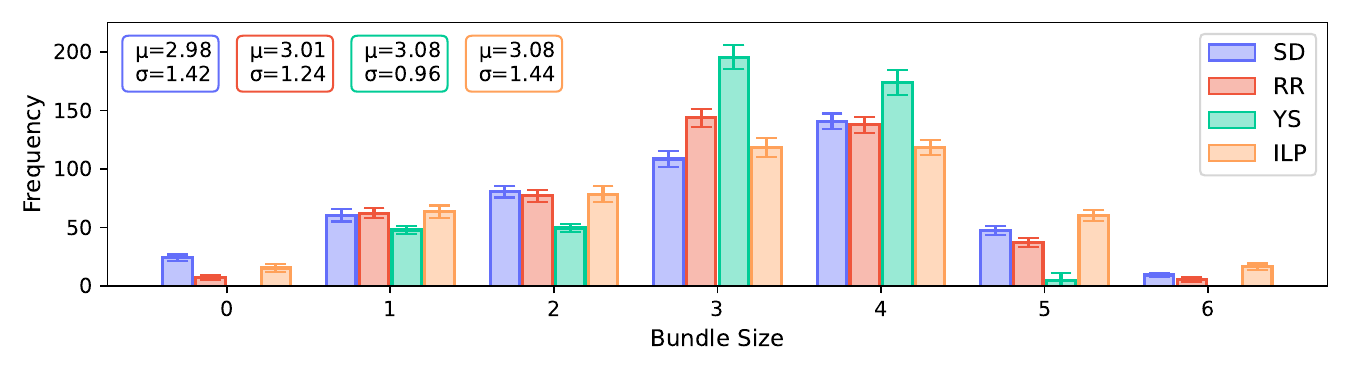}
    \caption{Histogram of the number of seats assigned to students under the four algorithmic benchmarks, for reduced student instances. We also report the mean $\mu$ number of seats assigned, as well as the standard deviation $\sigma$ of the assignment distribution for each allocation mechanism.}
    \label{fig:histogram-reduced}
        \Description{Results on reduced student cohort.}
\end{figure}
the distribution of student utilities (\Cref{fig:histogram-reduced}) shows that the allocations output by YS have ensure significantly more equitable distributions: most students receive bundles of 3 or 4 courses, while none receive 0 or 6 classes.  

We also run a real-scale experiment with the three scalable algorithms: SD, RR, and YS. 
For this, we use all real responses and augment them with synthetic students until reaching the actual quantities per status, resulting in a full-sized instance of 2,308 students (the number of CS majors at \UMass). 
We sample 100 different instances of synthetic students and run the algorithms in the same student order for each instance.

YS, SD and RR exhibit similar behavior on both welfare (\Cref{fig:welfare-metrics}) and fairness (\Cref{fig:fairness-metrics}) metrics. 
SD is the least fair and efficient of the three methods, YS offers the best guarantees, and RR falls somewhere in between.  

\begin{figure}[t]
        \centering
        \begin{subfigure}[b]{0.25\textwidth}
        \centering
        \includegraphics[width =\textwidth]{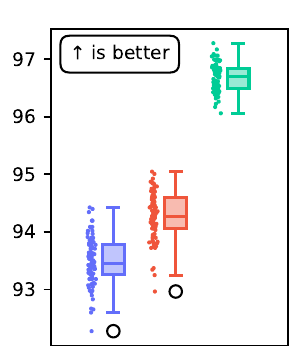}  
        \caption{\% of seats assigned}\label{fig:num-seats}
        \end{subfigure}
        \begin{subfigure}[b]{0.25\textwidth}
        \centering
        \includegraphics[width =\textwidth]{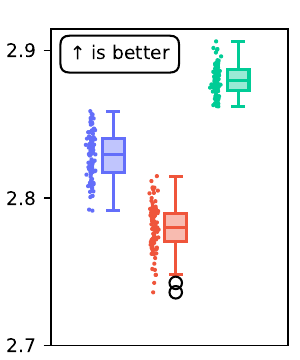}  
        \caption{Nash welfare}\label{fig:nash}
        \end{subfigure}
        \begin{subfigure}[b]{0.25\textwidth}
        \centering
        \includegraphics[width =\textwidth]{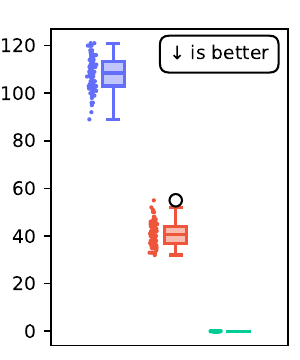}  
        \caption{Num. of $0$ util. agents}\label{fig:zeros}
        \end{subfigure}
        \caption{Welfare metrics for a full-size instance with 2,308 students.}\label{fig:welfare-metrics}
        \Description{Results on larger student cohort.}
\end{figure}
\begin{figure}[t]
        \centering
        \begin{subfigure}[b]{0.25\textwidth}
        \centering
        \includegraphics[width =\textwidth]{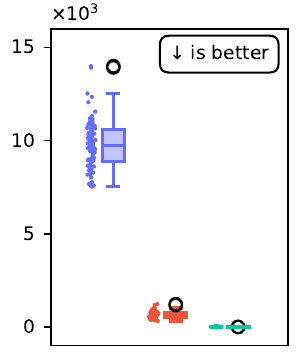}  
        \caption{PMMS violations}\label{fig:PMMS}
        \end{subfigure}
        \begin{subfigure}[b]{0.25\textwidth}
        \centering
        \includegraphics[width =\textwidth]{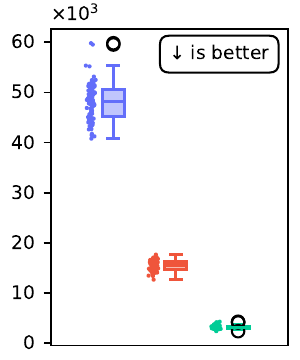}  
        \caption{EF violations}\label{fig:EF}
        \end{subfigure}
        \begin{subfigure}[b]{0.25\textwidth}
        \centering
        \includegraphics[width =\textwidth]{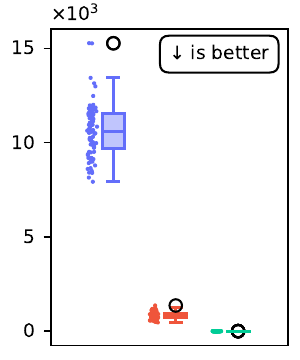}  
        \caption{EF-1 violations}\label{fig:EF1}
        \end{subfigure}
        % % \includegraphics[scale=0.5]{EC 2025/Figs/boxplot_PMMS_violations.pdf} 
        % % \includegraphics[scale=0.5]{EC 2025/Figs/boxplot_EF_violations.pdf} 
        % % \includegraphics[scale=0.5]{EC 2025/Figs/boxplot_EF1_violations.pdf}
        \caption{Fairness metrics for a full-size instance with 2,308 students}\label{fig:fairness-metrics}
        % \caption{Results for 100 runs of the Serial Dictatorship, Round Robin, and Yanke Swap algorithms on a real sized instance consisting of 2308 students created from real and synthetic responses. The top figure presents the six studied performance metrics: total allocated seats (USW), Nash welfare, number of empty bundles, and violations of PMMS, EF, and EF-1. The bottom figure shows a histogram of the average frequencies of allocated bundle sizes for each algorithm, including error bars, as well as the mean and standard deviation across all seeds.}
        % \label{fig:large_experiment}
        \Description{Results on larger student cohort.}
\end{figure}
\Cref{fig:hist} shows a histogram of bundle sizes for full-size instances. Again, both the serial dictatorship and the Round Robin mechanisms offer significantly greater disparities as compared to the Yankee Swap framework; indeed, as \Cref{fig:zeros} shows, $\sim 110$ students receive an empty bundle under SD, $\sim40$ receive an empty bundle under RR, and no  student receives an empty bundle under YS. Yankee Swap bundles are again heavily concentrated around 3-4 seats, which indicates a more equitable seat distribution.
\begin{figure}[t]
    \centering
    \includegraphics[width=\textwidth]{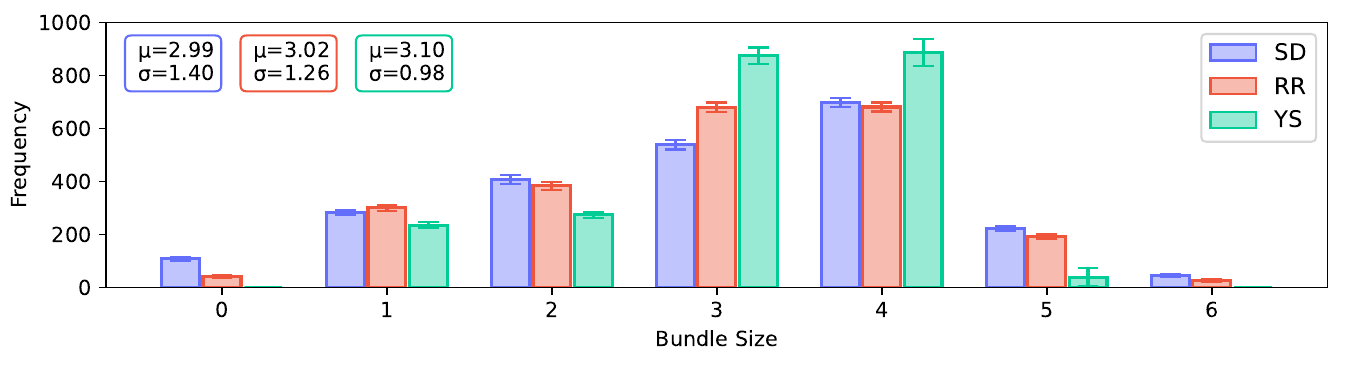}
    \caption{Averaged distribution of bundle sizes in experiments with 2,308 students.}
    \label{fig:hist}
    \Description{Student histogram full size instance}
\end{figure}

\section{Discussion}
\label{sec:discussion}
%effective surveying, repeated surveys, longitudenal
%student priorities
%the effects of scaling up
%uncertainty about final bundle - SD is better in that regard, how is this going to affect the result
%anticipating and shaping course demand
%effect of short transfer paths
In this work, we present a publicly available, large-scale framework including detailed student preference data, fair allocation mechanisms, and evaluation metrics. Our evaluation suite is not complete. In particular, we hope to include additional mechanisms in our suite (Course Match is a prime candidate) as well as additional survey data. We also plan to conduct longitudinal studies on the perceived effectiveness of different course allocation mechanisms within \UMass. 
Discussions with decision makers at \UMass identify an interesting advantage of serial dictatorship: students do not need to report their full preferences to the university, and have complete certainty regarding their course schedule when they conclude their interaction with the allocation mechanism. 
Assessing the impact of the inherent uncertainty of sequential allocation on student satisfaction is an important direction for future work.  

Assessing the effect of course \emph{supply and demand} is an important direction for future work. Preliminary analysis (see \Cref{appendix:experiments}) indicates that artificially increasing course demand by introducing additional students results in poorer performance on fairness metrics for all mechanisms; however, since Yankee Swap offers provable fairness guarantees, it results in more equitable outcomes than other benchmarks. For example, as we increase the number of students (while keeping course capacities constant) the number of students who receive no desired classes grows linearly under serial dictatorship. Since Yankee Swap outputs leximin allocations, this effect does not occur.

We would like to examine the effects of student \emph{uncertainty} on their enrollment behavior. Students' survey responses are (likely shaky) estimates of their true satisfaction with courses. 
In practice, students are aware of this, and regularly over-enroll in classes, with the intent of dropping some of them during first two weeks of the semester. 
This creates overheads for both the students and course enrollment systems, which could potentially be resolved with better algorithmic designs.
Pandora's Box problems \cite{beyhaghi2024pandora} model costly exploration under uncertainty, and have been used in other combinatorial optimization problems, e.g., minimum spanning trees \cite{singla2018priceinformation}. 
Applying similar techniques to the fair allocation domain is a promising direction for future work.  
% \paragraph{Acknowledgments} The authors thank Christine Holbrook, Ramesh Sitaraman for their support and insights in survey design, timing and data collection.  

Finally, we would like to encode more complex elicitation mechanisms and their effects on student satisfaction. In particular, explicitly eliciting students' preferences over course bundles (as is done by \citet{budish2017coursematch}), and integrating those preferences in our frameworks is an important direction for future work. The sequential frameworks we utilize need to be adapted in order to account for these changes. In particular, the fact that the underlying preferences explicitly encode synergies may require the integration of more complex sequential approaches, or their integration with optimization based methods.  
\subsection*{Acknowledgments} 
The authors thank the reviewers of the 1st Workshop on Incentives in Acadamia 2024 for their useful feedback on an early version of this work. 
Navarrete D{\'i}az and Zick are supported by NSF grant RI-2327057. 
Cousins was supported by a UMass Center for Data Science (CDS) postdoctoral fellowship. The authors thank David Thibodeau and Vignesh Viswanathan for their comments and feedback on early versions of this work. Finally, the authors would like to thank Christine Holbrook and Ramesh Sitaraman for their support and insights regarding survey design, release and data collection. 
\bibliographystyle{plainnat}

\bibliography{abb,references,george}
\newpage

% \section{Throwaway Section: What kinds of experiments are we running and what results we want}
% \begin{description}
%     \item[Are the allocations output efficient?]
%     I want a thorough comparison of efficiency across all welfare metrics: USW, Lex and NSW. How do these metrics change as a function of $n$, $m$ and capacities? How do they change as we modify the degree of competition among students and as the student/item ratio gets worse? 
%     \item[Are the allocations fair?]
%     How much do students envy each other? Does the MMS guarantee make sense here? Can we effectively compute it? What about the PMMS guarantee? Discuss the futility of MMS/Proportionality. 
%     \item[How long does it take to compute things?] Compare runtime of the different approaches and implementations of YS. Compare to other methods. 
% \end{description}

\appendix

\section{Synthetic Student Details}
\label{sec:synth_details}

In this section we elaborate on the \texttt{mBeta} distribution from modeling, to inference, and finally parameter estimation. Our emphasis is on providing enough detail to justify the implementation we have used in our experiments. In what follows, the \texttt{diag} operator constructs an $n \times n$ matrix with $\bm v$ along its diagonal when $\bm v$ is an $n \times 1$ vector, and it extracts the diagonal of $A$ into an $n \times 1$ vector when it operates on $n \times n$ matrix $A$. By $A \odot B$, we denote the Hadamard (or component-wise) product of matrices $A$ and $B$. The length-$m$ vector of all ones is denoted by $\bm{1}_m$, and the $m \times m$ identity matrix is given by $I_m$.

There are various ways to define a multivariate beta distribution for random student $\vec \vartheta$.  \citet{westphal2019simultaneous} constructs $\mBeta(\vec\gamma)$ from a Dirichlet distribution with concentration parameter $\vec\gamma \in \mathbb{R}^{2^m}$. 
In principle, $\vec\gamma$ can represent arbitrary dependencies between all possible subsets of the $m$ courses: The $\mBeta(\vec\gamma)$ distribution defines a categorical distribution over the power set of courses, and for $\vec p \sim \Dir(\vec\gamma)$, the matrix $H \in \{0, 1\}^{m \times 2^m}$, whose columns comprise all bit vectors of length $m$, maps $\vec p$ to $\vec\vartheta$, where $\vec \vartheta \sim \mBeta(\vec\gamma)$. 
In practice, we limit $\vec\gamma$ to size polynomial in $m$. \citet{westphal2019simultaneous} achieves this end by constraining $\vec\gamma$ to a family of distributions characterized by the mean vector $\E[\vec\vartheta] \in \R^m$ and covariance matrix $\cov(\vec\vartheta) \in \R^{m \times m}$ of $\mBeta(\vec \gamma)$. 

We approximate the preference of random student $\vec \vartheta$ for course $g$ to arbitrary precision via a sequence of  coin flips biased by $\theta_g$.
The entire preference vector $\vec \theta$ is encoded in the data matrix $\mathcal{D} \in \{0, 1\}^{\ell \times m}$, where $\mathcal{D}_{jg}$ represents the $j$th flip for course $g$.
% Aggregating coin flips for course $j$ across all $n$ students gives $Y^k_j = \sum_{i=1}^n X^k_{ij}$, where $Y^k_j \sim \Bin(n, \vartheta_j^k)$. 
%Note that, for fixed $\bm \vartheta^k = \theta^k$, each of the $n m$ coin flips in $\bm Y^k$ are pairwise independent. Moreover, since they share fixed parameter $\theta^k_j$, the entries of column $j$ of $\bm Y^k$ coin flips for the same course, $Y^k_{ij}$ are \emph{iid}, $\Bern(\theta_j)$. However, when $\bm \vartheta^k$ remains unobserved, each entry $Y^k_{ij}$ is distributed $\Bern(\vartheta^k)$, a hierarchical distribution. And since the components of $\bm \vartheta^k$ are dependent (by construction), the components of any given row of $\bm Y^k$ are also dependent.
Matrix $\mathcal D$ forms the input to our inference procedure for developing a statistical model of student preferences.
For computational reasons we also use a \emph{vectorized} version of $\mathcal{D}$: define the function $d: \{0,1\}^\ell \times \{0,1\}^m \rightarrow \mathbb{Z}^{2^m}$ to be a map from $\mathcal{D}$ to a vector whose $i$-th entry corresponds to the length-$m$ bit array with integer value $i$ and whose value at that entry is equal to the number of rows in $\mathcal{D}$ that have value $i$ when interpreted as a bit array. 

%For student $\bm \theta^k$ we model the survey response $y_j$ as a sample of random variable $Y_j \sim \Bin(n, \theta_j)$. Thus, we assume that the survey responses for any given student are marginally binomially distributed. We indirectly model the course quantity preference $k$ by scaling the responses $y_j$ so that $\frac1n\sum_{j = 1}^m y_j = k$. In particular, we normalize each entry of $\bm y$ by $\frac{k}{\sum_{j = 1}^m y_j}$. This approach ensures that responses from student $\bm \theta$, when converted to preferences and summed, will have expected value equal to $k$.
\begin{table}
\centering
\begin{tabular}{c c c c} 
 \hline
 Variable & Prior & Posterior Update & Description \\ [0.5ex] 
 \hline
 $\nu$ & $\frac{1}{1000}$ & $\nu + \ell$ & Shape parameter \\
  $\vec \mu$ & $\frac{1}{2} \bm{1}_m$ & $\texttt{diag}(A^*) / \nu^*$ & Mean vector \\ 
 $\vartheta_j$ & $\texttt{Beta}(\nu_j, \nu_j)$ & $\texttt{Beta}(\alpha_j, \beta_j)$ & Course preference values \\
 $\mathcal{R}$ & $I_m$ & $(V^*)^{-1/2} \Sigma^* (V^*)^{-1/2}$ & Correlation matrix \\
 $U$ & Proposition~\ref{prop:update_matrix} & --- & Update matrix \\ 
 $A$ & $\nu((\nu+\bm{1}_m) \Sigma + \vec\mu \vec\mu^T)$ & $A + U$ & Moment matrix \\
 $\vec \alpha$ & \multicolumn{2}{c}{$\nu \vec \mu$} & Shape parameter \\
 $\vec \beta$ & \multicolumn{2}{c}{$\nu \bm{1}_m - \vec \alpha$} & Shape parameter \\
 $V$ & \multicolumn{2}{c}{$\texttt{diag}(\vec\mu \odot (\bm{1}_m-\vec\mu)) / (\nu + 1)$} & Standard deviation matrix \\
 $\Sigma$ & \multicolumn{2}{c}{$V^{1/2}\mathcal{R}V^{1/2}$}  & Covariance matrix \\
 \hline
\end{tabular}
\caption{Synthetic student inference parameters. Variables with asterisks indicate posterior values. Expressions spanning both prior and posterior columns are used in both cases, with prior input variables being used in the former and posterior input variables used in the latter.}
\label{tab:inference}
\end{table}

We construct a separate $\mBeta(\vec \gamma)$ distribution for each respondent from their survey responses by using the following inference procedure. 
For respondent $i$, we use the data matrix $\mathcal D^i$ for a Bayesian update of the prior distribution for $\mBeta(\vec \gamma)$. \Cref{tab:inference}
provides a summary of the prior and posterior parameters whose values we derive presently. 

Let $\vec\mu$ denote any prior estimate of $E[\vec\vartheta]$ and assume for the moment that prior parameters $\nu \in \mathbb{R}_+$ and $\mathcal{R} \in (-1,1)^{m \times m}$ are given. Westphal~\cite{westphal2019simultaneous} expresses the covariance $\texttt{cov}(\vec\vartheta)$ by 
\begin{equation}
    \label{eq:sigma}
    \Sigma = V^{1/2}\mathcal{R}V^{1/2},
\end{equation} 
where
\begin{equation}
\label{eq:V}
    V = \texttt{diag}(\vec\mu \odot (\bm{1}_m-\vec\mu)) / (\nu + 1).
\end{equation}
(Note that Equation~\ref{eq:V} is slightly different than what is found in \citet{westphal2019simultaneous} due to a typographical error in the latter.) The \emph{moment matrix} becomes
\begin{equation}
\label{eq:moment_matrix}
    A = \nu((\nu+1) \Sigma + \vec\mu \vec\mu^T).
\end{equation}

It's important to select the priors $\vec\mu$, $\nu$, and $\mathcal{R}$ such that Westphal's moment conditions are met. To that end, we may choose any $\nu \in \mathbb{R}_+$, $\vec\mu = \bm{1}_m/2$, and $\mathcal{R} = I_m$, the $m$-dimensional identity matrix. For the marginal prior with respect to course $g$ we have
\begin{equation*}
    \vartheta_g = \Beta(\nu/2, \nu/2).
\end{equation*}

Posterior updates are $\nu^* = \nu + \ell$ and $A^* = A + U$ where $U = H \Lambda H^T$, $\Lambda = \texttt{diag}(d(\mathcal D))$, $\mathcal D$ is the data matrix, and $d(\mathcal D)$ its vectorization as defined above. Posterior mean vector $\vec\mu^*$ is updated from $A^*$ and $\nu^*$ via the relationship
\begin{equation*}
    \vec\mu^* = \texttt{diag}(A^*) / \nu^*.
\end{equation*}
The posterior $\Sigma^*$ can be derived from Equation~\ref{eq:moment_matrix} after substituting $A^*$, $\vec\mu^*$, and $\nu^*$ for $A$, $\vec\mu$, and $\nu$, respectively. Then $\mathcal{R}^*$ can be derived from Equations~\ref{eq:sigma} and~\ref{eq:V} after substituting $\Sigma^*$ and $V^*$ for $\Sigma$ and $V$, respectively. Specifically,
\begin{equation*}
    \mathcal{R}^* = (V^*)^{-1/2} \Sigma^* (V^*)^{-1/2}.
\end{equation*}
Finally, the marginal posterior with respect to course $g$ is given by
\begin{equation}
\label{eq:theta_posterior}
    \vartheta_g = \Beta(\alpha_g, \beta_g),
\end{equation}
where $\vec \alpha = \nu \vec \mu$ and $\vec \beta = \nu \bm{1}_m - \vec \alpha$.

The following proposition is necessary for calculating $U$ without explicitly calculating $H$, the latter of which scales exponentially with the number of classes.

\begin{proposition}
\label{prop:update_matrix}
    Let $d_i$ denote the $i$th entry of vector $d(\mathcal D)$, $\Lambda = \texttt{diag}(d(\mathcal D))$, and take $\vec y_i$ to be the $i$th column of $H$. Then we have
    \begin{equation*}
        U = \sum_{i = 1}^{2^m} d_i \vec y_i \vec y_i^T.
    \end{equation*}
\end{proposition}

\begin{proof}
Let $w = 2^m$ and take $E_i$ to denote the matrix with all zeros except for a 1 in the $i$th diagonal entry. Then it is clear that
\begin{equation*}
    \Lambda = \sum_{i=1}^w d_i E_i,
\end{equation*}
and so
\begin{equation}
\label{eq:inter_U}
    U = \sum_{i=1}^w d_i H E_i H^T. 
\end{equation}
We also have that $H E_i H^T$ is equivalent to
\[
\begin{array}{c}
    [\vec y_1, \ldots, \vec y_w]
    E_i
    \left[
    \begin{array}{c}
         \vec y_1^T  \\
         \vdots \\
         \vec y_w^T
    \end{array}
    \right] =
\end{array}
\]
\[
\begin{array}{c}
    [0, \ldots, 0, \vec y_i, 0, \ldots, 0]
    \left[
    \begin{array}{c}
         \vec y_1^T  \\
         \vdots \\
         \vec y_w^T
    \end{array}
    \right] =
\end{array}
\]
\[
[y_{i1} \vec y_i , \ldots, y_{iw} \vec y_i ]=
\]
\begin{equation}
\label{eq:outer_product}
\vec y_i \vec y_i^T.   
\end{equation}
The result follows by substituting Equation~\ref{eq:outer_product} into Equation~\ref{eq:inter_U}.
\end{proof}

The \mBeta distribution is defined so that the $m$ marginals have a beta distribution. 
Let $C: [0,1]^m \rightarrow [0,1]$ be the Gaussian copula with correlation matrix equal to the posterior correlation matrix $\mathcal{R}$. The Gaussian copula~\cite{sklar1959fonctions} is a multivariate cumulative distribution function that captures the dependency structure of the multivariate Gaussian but has uniform marginals. It allows us to \emph{borrow} the dependency structure of a multivariate Gaussian while separately defining each marginal distribution. We recover the distribution function $F$ for $\mBeta(\vec\gamma)$ as 
\begin{equation*}
\label{eq:mbeta_dist}
    F = C(\tilde{F}_1, \ldots, \tilde{F}_m),
\end{equation*}
where $\tilde{F}_g$ is the distribution function of the marginal posterior distribution for $\vartheta_g$ given by Equation~\ref{eq:theta_posterior}. Given $\mBeta(\vec \gamma)$ (where $\vec \gamma$ is implicitly defined), we produce a synthetic survey response by drawing a synthetic student $\vec \sigma \in [0,1]^m$ according to $\mBeta(\vec \gamma)$, multiplying each of its entries $\sigma_g$ by 8 (the highest possible response value), and rounding it the nearest integer. %The maximum number of courses is given by $\bm \theta^T \bm 1$, the sum preferences $\bm \theta$.

\section{Empirical Transfer Path Length}
\label{appendix:transfer_paths}
As explained in \Cref{sec:yswithmult}, Yankee Swap requires finding and executing transfer paths, the length of which govern the algorithm's runtime (see \Cref{subsec:runtime}). We analyze the lengths of these paths across 100 seeds of the real-scale instance with 2308 students. \Cref{fig:transfer_paths} presents a histogram showing the average frequencies of path lengths, along with their respective error bars. A path length of 0 indicates that no path was found, at which point the agent was removed from the game; a length of 1 means the agent directly received a seat without needing to steal. Paths of length 2 or more represent cases where transfers occurred between students. The histogram reveals that most iterations involve either 0 or 1 students, indicating that transfers are rare. This is particularly surprising given that these very few transfers are responsible for YS outperforming other algorithms in welfare and fairness metrics. Lastly, although there is no theoretical upper bound on transfer path length --- in the worst case, it could involve all $m$ items --- we observe that even across 100 simulations, each averaging 9500 iterations, the longest path recorded is 7 (see \Cref{tab:transfer_path}).

\begin{figure}[h]
        \centering
        \includegraphics[width =\textwidth]{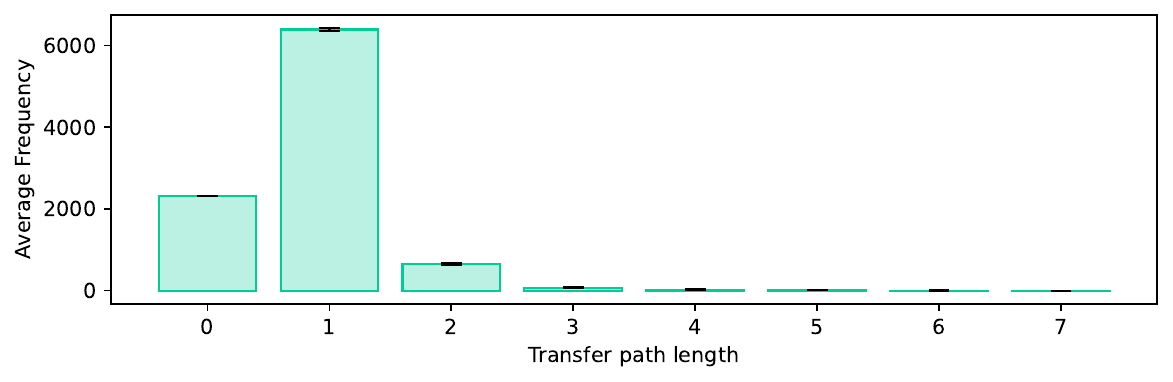}  
        \caption{Transfer path length mean frequency across all Yankee Swap iteration for 100 seeds.}
        \Description{Transfer path length mean frequency}
        \label{fig:transfer_paths}
\end{figure}

\section{Additional Experiments}
\label{appendix:experiments}
How well do different mechanisms respond to increased student demand? 
To evaluate how different allocation mechanisms behave under increased student demand, we conducted 10 simulations per cohort size while maintaining real course capacities. 
This allows us to analyze what happens when we increase the number of students without adding more course seats. 
As shown in \Cref{fig:stress_USW}, all three algorithms perform well in terms of utilitarian welfare, with YS slightly outperforming the others in terms of \USW. 

However, a different trend emerges when considering empty bundles (\Cref{fig:stress_zeroes}). 
YS makes a significant effort to ensure every student receives at least one item, as reflected in the curve’s shape, which suggests that YS prioritizes allocating each student at least one desirable class, until demand completely outstrips supply. Given that there is a total of 7,389 seats, enrolling more than 7,389 students inevitably leads to empty bundles. 
While RR also mitigates empty bundles, SD exhibits a linear increase in empty bundles as system stress rises.

\begin{table}
    \centering
    \small
    \begin{tabularx}{\textwidth}{cXXXXXXXX}
    \toprule
         Transfer path length & 0 &1 & 2 & 3 & 4 & 5 & 6 & 7\\
         \toprule
         Frequency & 115,400 & 319,539 & 32,082 & 3,583 & 1,253 & 580 & 83 & 4\\
         \bottomrule
    \end{tabularx}
    \caption{Frequency of each transfer path length across all Yankee Swap iterations for the full-sized instance, for all 100 seeds.}
    \label{tab:transfer_path}
\end{table}

\begin{figure}[t]
        \centering
        \includegraphics[width =\textwidth]{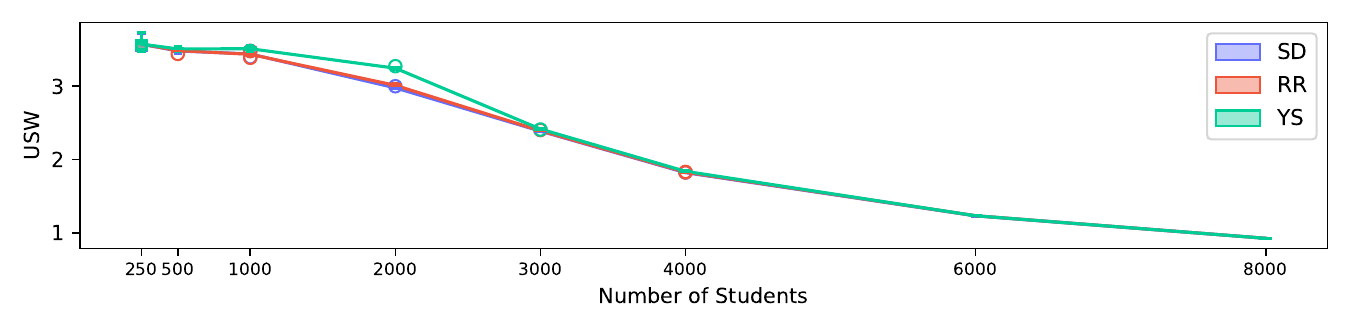}  
        \caption{\USW as a function of the number of students, while maintaining real course capacities.}
        \Description{USW as a function of the number of students.}
        \label{fig:stress_USW}
\end{figure}

\begin{figure}[t]

        \centering
        \includegraphics[width =\textwidth]{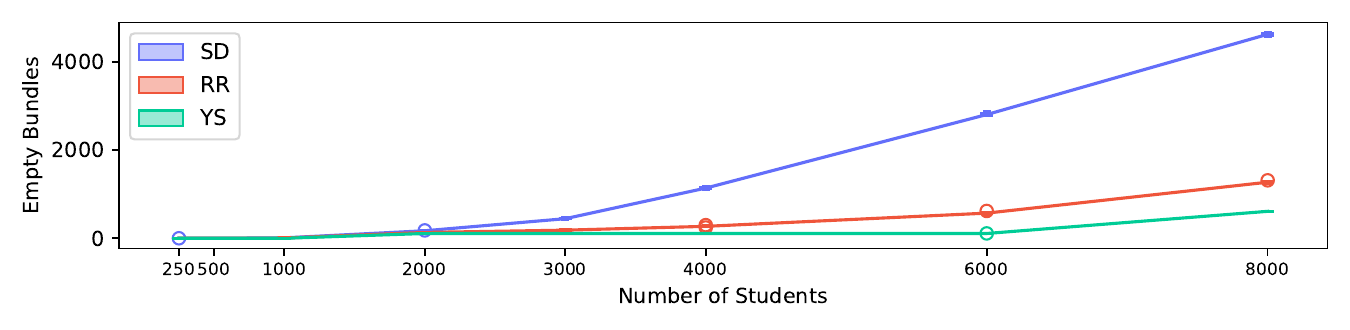}  
        \caption{Number of empty bundles as a function of the number of students, while maintaining real course capacities.}
        \Description{Empty bundles as a function of the number of students.} \label{fig:stress_zeroes}
\end{figure}

\end{document}